\documentclass[a4paper]{article}
\usepackage[T1]{fontenc}
\usepackage[utf8]{inputenc}
\usepackage{authblk}
\usepackage[dvipsnames]{xcolor}
\usepackage{a4wide}
\usepackage{graphicx}
\usepackage{amsmath}
\usepackage{amssymb}
\usepackage{cases}
\usepackage{hyperref, graphicx}
\usepackage{enumitem}
\usepackage{tikz}
\usetikzlibrary{automata, arrows.meta, positioning,calc}

\usepackage[vlined,english,boxed,linesnumbered]{algorithm2e}
      \SetAlgoLined
      \SetKwInput{Input}{Input}
      \SetKwInput{Output}{Output}
      \SetKw{Return}{Return}
      \SetKwIF{If}{ElseIf}{Else}{If}{then}{else if}{else}{endif}
      \SetKwFor{For}{For}{do}{done}
      \SetKwFor{for}{For}{}{}
      \SetKwFor{While}{While}{do}{done}
      \SetKwFor{Repeat}{Repeat}{}{}
      \SetKwRepeat{Do}{Do}{while}
      \SetKwFor{Procedure}{Procedure}{}{}
      \SetKwFor{Function}{Function}{}{}
      \SetKw{Return}{Return}%
      \SetKw{return}{return}%
      \SetKwComment{Comment}{/* }{ */}
      \SetSideCommentRight
      \DontPrintSemicolon

\usetikzlibrary{positioning}
\usetikzlibrary{graphs,graphs.standard,quotes}

\usepackage{amsmath}

\usepackage{amsthm}

\usepackage{subcaption}

\newtheorem{theorem}{Theorem}
\newtheorem{lemma}{Lemma}
\newtheorem{proposition}{Proposition}

\let\bottom=\perp
\newcommand{\card}[1]{\left|{#1}\right|}
\newcommand{\set}[1]{\left\{{#1}\right\}}
\DeclareMathOperator{\vertexsorting}{VertexSorting}
\DeclareMathOperator{\extremalnodes}{ExtremeNodes}
\DeclareMathOperator{\triangleextremities}{TriangleExtremities}

\renewcommand{\figurename}{Figure}

\title{Forbidden Patterns in Temporal Graphs Resulting from Encounters in a Corridor}

\author[1]{M\'onika Csik\'os}
\author[1]{Michel Habib}
\author[1]{Minh-Hang Nguyen}
\author[1]{Mikaël Rabie}
\author[2]{Laurent Viennot}
\affil[1]{Université Paris Cité, CNRS, IRIF, F-75013, Paris, France}
\affil[2]{Inria, DI ENS, Paris, France}
\begin{document}

\maketitle

\begin{abstract}

  In this paper, we study temporal graphs arising from mobility models, where vertices correspond to agents moving in space and edges appear each time two agents meet. We propose a rather natural one-dimensional model.

If each pair of agents meets exactly once, we get a simple temporal clique where the edges are ordered according to meeting times. 
  In order to characterize which temporal cliques can be obtained as such `mobility graphs', we introduce the notion of forbidden patterns in temporal graphs. Furthermore, using a classical result in combinatorics, we count the number of such mobility cliques for a given number of agents, and show that not every temporal clique resulting from the 1D model can be realized with agents moving with different constant speeds. For the analogous circular problem, where agents are moving along a circle, we provide a characterization via circular forbidden patterns. 

   Our characterization in terms of forbidden patterns can be extended to the case where each edge appears at most once. 
    We also study the problem where pairs of agents are allowed to cross each other several times, using an approach from automata theory. We observe that in this case, there is no finite set of forbidden patterns that characterize such temporal graphs and nevertheless give a linear-time algorithm to recognize temporal graphs arising from this model.

\end{abstract}


\section{Introduction}

\subsection{Motivation}
Temporal graphs (also known as dynamic, evolving or time-varying networks) can be informally described as graphs that change with time. Their study is the subject of many theoretical and practical research works as various   real-world systems that can be modelled with temporal graphs (such as social contacts, co-authorship graphs, or transit networks), see \cite{XuanFJ2003,CasteigtsFQS2012,holme2012temporal,LatapyVM2018,michail2016introduction}. 
A very natural range of models for temporal graphs comes from mobility. When agents move around in space, we can track the moments when they meet each other and obtain a temporal graph. The main objective of this work is to characterize temporal graphs resulting from certain mobility models.

A classical model used for mobile networks is the unit disk graph, where vertices correspond to unit disks in the plane, and two disk are adjacent if and only if they intersect. When the disks are allowed to move, we obtain a so-called dynamic unit disk graph~\cite{VillaniC2021}, and the appearance of edges then forms a temporal graph. We consider a one-dimensional version where the disks are moving along a line or equivalently a narrow corridor of unit width. This could encompass practical settings such as communicating cars on a single road. In particular, if each car has constant speed, 
each pair of cars encounters each other at most once. We further restrict to the sparse regime where each disk intersects at most one other disk at a time. In other words, the edges appearing at any given time always form a matching. This restriction, called local injectivity, has already been considered in the study of simple temporal cliques~\cite{casteigts2021temporal} which are temporal graphs where an edge between any pair of nodes appears exactly once.

If two agents can communicate when they meet, one can ask how information can flow in the network. The appropriate notion of connectivity then arises from temporal paths which are paths where edges appear one after another along the path. The problem of temporal connectivity has been considered, by Awerbuch and Even~\cite{awerbuch1984efficient}, and studied more systematically in \cite{kempe2000connectivity}. A temporal spanner can be defined as a temporal subgraph that preserves connectivity. An interesting question concerning temporal graphs is to understand which classes of temporal graphs have temporal spanners of linear size. Although some temporal graphs have only $\Theta(n^2)$-size temporal spanners~\cite{AxiotisF16}, simple temporal cliques happen to have $O(n\log n)$-size temporal spanners~\cite{casteigts2021temporal}, and it is conjectured that temporal cliques could even have linear size spanners.
Indeed, a natural question is whether temporal graphs resulting from a mobility model can have sparse spanners. In particular, do temporal cliques arising from our 1D model have temporal spanners of linear size?

\subsection{Our contribution}

Our main contribution is a characterization of the simple temporal cliques that result from this 1D model. A simple temporal clique can only arise when agents start out in a certain order along the corridor and end up in the opposite order, after each pair of agents crossing exactly once. We provide a characterization of such ``1D-mobility'' temporal cliques in terms of forbidden  ordered patterns on three nodes. 
As far as we know, we introduce the first definition of forbidden patterns in a temporal graph. Our definition is based on the existence of an order on the nodes (which actually corresponds to their initial order along the line). A forbidden pattern is a temporal subgraph with a relative ordering of its nodes, and with a forbidden relative ordering of its edges according to their time labels.

Our characterization of 1D-mobility temporal cliques  leads directly to an $O(n^3)$-time algorithm for testing whether an ordering of the $n$ nodes of a temporal clique is appropriate and allows to exclude these patterns. 
Interestingly, an $O(n^2)$-time algorithm allows to find such an appropriate initial ordering of the nodes from the list of the edges of the clique ordered by appearance time. Moreover, we  can actually check in $O(n^2)$ time that this order excludes the forbidden patterns to obtain an overall linear-time recognition algorithm, since we have $n(n-1)/2$ edges in our graphs.

Another way of looking at this problem is sorting through adjacent transpositions an array $A$, where $n$ elements are initially stored in reverse order. At each step, we choose an index $i$ such that $A[i]>A[i+1]$ and swap the two elements at positions $i$ and $i+1$. The array is guaranteed to be sorted in increasing order after $T=n(n-1)/2$ steps, since the permutation of the elements in $A$ has initially $T$ inversions while each step decreases this number by one. Note that this is reminiscent of bubble sorting, which indeed operates according to a sequence of such transpositions.
This naturally connects our 1D model to the notion of reduced decompositions of a permutation~\cite{tenner2006reduced}. Using this observation and a classical combinatorial result, we give a formula for the number of temporal cliques with $n$ nodes resulting from our 1D model.

 In addition, we show that our temporal cliques do contain temporal spanners of linear size (with exactly $2n-3$ edges) by highlighting a pertinent temporal subgraph that considers only the edges incident to one of the two extreme agents in the initial order along the line.

Finally, we consider some generalizations. First, we study those 1D mobility graphs where each edge appears \emph{at most} once and provide a forbidden pattern characterization of this setup as well. 
Second, we consider what might be a forbidden pattern definition if edges can occur multiple times, that is when some  pairs of agents can cross each other multiple times.

In the case of 1D model when agents move with different constant speeds, we observe that the resulting temporal mobility cliques are not characterized by the same set of forbidden patterns. 

Finally, we consider some generalised models. The first one is the 1D model in which each pair of agents meet at most once. We provide an extended set of forbidden patterns that characterizes temporal graphs arising in this situation. Our second extension of the 1D model is when each pair of agents can meet several times, that is, the arising temporal graphs can have multi-edges. We give a linear-time algorithm to recognise these multi-crossing 1D mobility temporal graphs. 
We also show that there is no finite set of forbidden patterns that characterizes these multi-crossing 1D temporal mobility graphs.
On the positive side, we give a characterization  using an automaton where states correspond to possible orderings between three nodes and where each possible crossing induces a transition to a new state.

The final extension is a circular model where agents move along a circle. We give forbidden patterns to characterize temporal cliques arising from this model when each pair of agents meet exactly once. Similarly to the linear case, the automata approach is used in the case where each pair of agents can meet several time.

\color{black}

\subsection{Related works}

\paragraph{Dynamic unit disk graphs} A closely related work concerns the detection of dynamic unit disk graphs on a line~\cite{Villani2021,VillaniC2021}. An algorithm is proposed to decide whether a continuous temporal graph can be embedded in the line along its evolution, such that the edges present at each time instant correspond to the unit disk graph within the nodes according to their current position in the embedding at that time. The sequence of edge events (appearance or disappearance) is processed online one after another, relying on a PQ-tree to represent all possible embeddings at the time of the current event according to all events seen so far. It runs within a logarithmic factor from linear time.
Our model is tailored for discrete time and assumes that two nodes cross each other when an edge appears between them. This is not the case in the dynamic unit disk graph model: an edge can appear during a certain period of time between two nodes even if they don't cross each other. The PQ-tree approach can probably be adapted to our model for a more general recognition of the temporal graphs it produces. However, our characterization leads to a faster linear-time algorithm for recognizing simple temporal cliques arising from our model.

\paragraph{Forbidden patterns} Since the seminal papers of Damaschke~\cite{damaschke1990forbidden} and Skrien~\cite{skrien1982relationship}, many hereditary graph classes have been characterized by the existence of an order of the vertices that avoids some pattern, i.e. an ordered structure. These include bipartite graphs, interval graphs, chordal graphs, comparability graphs and many others. In \cite{hell2014ordering}, it is proved that any class defined by a set of forbidden patterns on three nodes can be recognized in  $O(n^3)$ time. This was later improved in~\cite{feuilloley2021graph} with  a full characterization of the 22 graph classes that can be defined with forbidden patterns on three nodes. 
An interesting extension to forbidden circular structures is given in \cite{Guzman-ProHH23}. The growing interest in forbidden patterns in the study of hereditary graph classes is partly supported by the certificate that such an ordering provides as checking if it avoids the patterns can be checked in polynomial time.

\paragraph{Reduced decompositions}  The number of reduced decompositions of a permutation of $n$ elements is studied in \cite{stanley1984number}. An explicit formula is given for the reverse permutation $n,n-1,\ldots,1$ based on the hook length formula~\cite{bandlow2008elementary,frame1954hook}. 

\subsection{Roadmap}

In Section~\ref{sec:def}, we introduce the key notions of the paper. In particular, we provide the precise definitions of temporal graphs, 1D mobility models, and forbidden temporal patterns. Section~\ref{sec:charac} contains our main results: a characterization of simple 1D mobility cliques through forbidden patterns, a formula for the number of such cliques of a given size, a detection algorithm, the description of a linear size spanner, and a discussion on the case of constant speed agents along the line. 
Section~\ref{sec:atmost} handles the case where each pair crosses at most once, by providing the forbidden pattern characterization. Section~\ref{sec:multiedges} considers the case where pairs can cross each other several times, for which we present a linear-time recognition algorithm and describe a related automaton. In Section~\ref{sec:circular}, we study the setup where agents are moving along a circle. We give a  forbidden pattern characterization for the case where pairs cross each other exactly once, and apply the automata approach for the case where pairs can cross each other several times.
 Finally, we state some open questions and perspectives in Section~\ref{sec:conclusion}.
\section{Preliminaries and mobility model}\label{sec:def}

In this section, we introduce the definitions and notations we will use through the paper. In particular, we  first define formally temporal graphs and forbidden patterns. We then introduce the mobility model and related combinatoric concepts.

\subsection{Temporal graphs and forbidden patterns}

Informally, a temporal graph is a graph with a fixed vertex set and whose edges change with time. A \emph{temporal graph} 
can be formally defined as a pair $\mathcal{G}=(G,\lambda)$ where $G=(V,E)$ is a graph with vertex set $V$ and edge set $E$, and $\lambda: E \rightarrow 2^{\mathbb{N}}$ is a labeling assigning to each edge $e\in E$ a non-empty set $\lambda(e)$ of discrete times when it appears.
We note $uv\in E$ the edge between the pair of vertices (or nodes) $u$ and $v$. 
 If $\lambda$ is locally injective in the sense that adjacent edges have disjoint sets of labels, then the temporal graph is said to be \emph{locally injective}. If $\lambda$ is additionally \emph{single valued} (i.e. $|\lambda(e)|=1$ for all $e\in E$), then $(G,\lambda)$ is said to be simple~\cite{casteigts2021temporal}. The maximum time label $T=\max \cup_{e\in E}\lambda(e)$ of an edge is called the \emph{lifetime} of $(G,\lambda)$. In the sequel, we will mostly restrict ourselves to simple temporal graphs and even require the following restriction of locally injective. A temporal graph is \emph{incremental} if at most one edge appears in each time step, that is $\lambda(e)\cap \lambda(f)=\emptyset$ for any distinct $e,f\in E$. 
 The reason for this restriction is mainly to ease the definition of isomorphism introduced later in this section. 
A (strict) \emph{temporal path} is a sequence of 
triplets $(u_i,u_{i+1},t_i)_{i\in [k]}$ such that $(u_1,\ldots,u_{k+1})$ is a path in $G$ with strictly increasing time labels: formally, for all $i\in [k]$, we have $u_i u_{i+1}\in E$, $t_i\in\lambda(u_i u_{i+1})$ and $t_i<t_{i+1}$.
Note that our definition corresponds to the classical strict version of temporal path as we require time labels to strictly increase along the path\footnote{The interested reader can check that the two notions of strict temporal path and non-strict temporal path are the same in locally injective temporal graphs.}.
A temporal graph is \emph{temporally connected} if every vertex can reach any other vertex through a temporal path. A \emph{temporal subgraph} $(G',\lambda')$ of a temporal graph $(G,\lambda)$ is a temporal graph such that $G'$ is a subgraph of $G$ and $\lambda'$ satisfies $\lambda'(e)\subseteq \lambda(e)$ for all $e\in E'$. A \emph{temporal spanner} of $\mathcal{G}$ is a temporal subgraph $\mathcal{H}$ preserving temporal connectivity, that is there exists a temporal path from $u$ to $v$ in $\mathcal{H}$ whenever there exists one in $\mathcal{G}$.
Given a subset $S\subseteq V$ of nodes, the temporal subgraph \emph{induced} by $S$ is defined as the temporal subgraph $\mathcal{G}_{|S}=((S,E’),\lambda’)$ of  $\mathcal{G}$ such that $E’=E\cap {S \choose 2}$ and $\lambda’(e)=\lambda(e)$ for all $e\in E’$. A temporal subgraph $\mathcal{G’}$ is said to be \emph{induced} if there exists a set $S$ such that $\mathcal{G’}=\mathcal{G}_{|S}$.

A representation $\mathcal{R}$ of a temporal graph $\mathcal{G}=((V,E),\lambda)$ is defined as an ordered list of $M=|\lambda|=\sum_{e\in E}|\lambda(e)|$ triplets $\mathcal{R}=(u_1,v_1,t_1),\ldots,(u_M,v_M,t_M)$ where each triplet $(u_i,v_i,t_i)$ indicates that edge $u_iv_i$ appears at time $t_i$. We additionally require that the list is sorted by non-decreasing time. In other words, we have $\lambda(uv)=\{ t : (u,v,t)\in \mathcal{R}\}$ for all $uv\in E$.
Note that any incremental temporal graph $\mathcal{G}$ has a unique representation denoted by
$\mathcal{R}(\mathcal{G})$.
Indeed, its temporal connectivity only depends on the ordering in which edges appear, we can thus assume without loss of generality that we have $\cup_{e\in E}\lambda(e)=[T]$ where $T$ is the lifetime of $((V,E),\lambda)$ (we use the notation $[T]=\{1,\ldots,T\}$). 
Given two incremental  temporal graphs $\mathcal{G}=((V,E),\lambda)$ and $\mathcal{G'}=((V',E'),\lambda')$, an \emph{isomorphism} from $\mathcal{G}$ to $\mathcal{G'}$ is a one-to-one mapping $\phi:V\to V'$ such that, for any $u,v\in V$, $uv\in E\Leftrightarrow \phi(u)\phi(v)\in E'$ ($\phi$ is a graph isomorphism), and their representation $\mathcal{R}(\mathcal{G})=(u_1,v_1,t_1),\ldots,(u_M,v_M,t_M)$ and $\mathcal{R}(\mathcal{G'})=(u'_1,v'_1,t'_1),\ldots,(u'_M,v'_M,t'_M)$ have same length $M=|\lambda|=|\lambda'|$ and are temporally equivalent in the sense that edges appear in the same order: $u'_iv'_i=\phi(u_i)\phi(v_i)$ for all $i\in [M]$.
When such an isomorphism exists, we say that $\mathcal{G}$ and $\mathcal{G'}$ are \emph{isomorphic}.
Given an integer  $L\le M$, we define the \emph{prefix of length $L$} of $\mathcal{G}$ as the temporal graph with representation $(u_1,v_1,t_1),\ldots,(u_L,v_L,t_L)$. We also say that $\phi$ is a \emph{prefix-isomorphism} from $\mathcal{G}$  to $\mathcal{G’}$ when $M'\le M$ and $\phi$ is an isomorphism from the prefix of length $M'$ of $\mathcal{G}$ to $\mathcal{G’}$. When such a prefix-isomorphism exists, we say that $\mathcal{G}$ is \emph{prefix-isomorphic} to $\mathcal{G’}$.

A \emph{temporal clique} is a temporal graph $(G,\lambda)$ where the set of edges is complete, and where we additionally require the temporal graph to be incremental and $\lambda$ to be single valued.
Notice that it is a slight restriction compared to the definition of \cite{casteigts2021temporal} which requires $(G,\lambda)$ to be locally injective rather than incremental. However, we do not lose in generality as one can easily transform any locally injective temporal graph into an incremental temporal graph with same temporal connectivity (we simply stretch time by multiplying all time labels by $n^2$ and arbitrarily order edges with same original time label within the corresponding interval of $n^2$ time slots in the stretched version).
With a slight abuse of notation, we then denote the label of an edge $uv$ by $\lambda(uv)\in\mathbb{N}$.

\smallskip

A \emph{temporal pattern} is defined as an incremental temporal graph $\mathcal{H}=(H,\lambda)$. An incremental temporal graph $\mathcal{G}=(G,\lambda')$ \emph{excludes} $\mathcal{H}$ when it does not have any temporal subgraph $\mathcal{H'}$ which is isomorphic to $\mathcal{H}$. A temporal pattern \emph{with forbidden edges} is a temporal pattern $\mathcal{H}=(H,\lambda)$ together with  a set $F\subseteq V\times V\setminus E$ of forbidden edges in $H=(V,E)$. An incremental temporal graph $\mathcal{G}=((V',E'),\lambda')$ \emph{excludes} $\mathcal{H}$ when it does not have any temporal subgraph $\mathcal{H'}$ which is isomorphic to $(H,\lambda')$ through an isomorphism $\phi$ respecting non-edges, that is any pair of nodes $u,v\in V'$ which is mapped to a forbidden edge $\phi(u)\phi(v)\in F$, we have $uv\notin E'$.
An incremental temporal graph $\mathcal{G}$ \emph{excludes} $\mathcal{H}$ as a \emph{prefix} when it does not have any induced temporal subgraph $\mathcal{H'}$ which is prefix-isomorphic to $\mathcal{H}$. In such case, we say that $\mathcal{H}$ is a \emph{forbidden prefix-pattern} of $\mathcal{G}$.

An \emph{ordered temporal graph} is a pair $(\mathcal{G},\pi)$, where $\mathcal{G}$ is a temporal graph and $\pi$ is an ordering  of its nodes. Similarly, an \emph{ordered temporal pattern}  $(\mathcal{H},\pi)$ is a temporal pattern $\mathcal{H}$ together with an ordering $\pi$ of its nodes. An ordered incremental temporal graph $(\mathcal{G},\pi')$ 
\emph{excludes} $(\mathcal{H},\pi)$ when it does not have any temporal subgraph $\mathcal{H'}$ which is isomorphic to $\mathcal{H}$ through an isomorphism $\phi$ preserving relative orderings, that is $\pi(\phi(u))<\pi(\phi(v))$ whenever $\pi'(u)<\pi'(v)$. We then also say that the ordering \emph{$\pi'$ excludes $(\mathcal{H},\pi)$ from $\mathcal{G}$}, or simply excludes $(\mathcal{H},\pi)$ when $\mathcal{G}$ is clear from the context.
We also define an ordered temporal pattern with forbidden edges similarly as above. 
We also say that an ordered incremental temporal graph $(\mathcal{G},\pi')$ \emph{excludes} $(\mathcal{H},\pi)$ as an \emph{ordered prefix} when it does not have any induced temporal subgraph $\mathcal{H'}$ which is prefix-isomorphic to $\mathcal{H}$ through an isomorphism $\phi$ preserving relative orderings, that is $\pi(\phi(u))<\pi(\phi(v))$ whenever $\pi'(u)<\pi'(v)$. In such case, $(\mathcal{H},\pi)$ is a \emph{forbidden ordered prefix-pattern} of $(\mathcal{G},\pi')$.

\subsection{1D-mobility model}

We introduce here the notion of a temporal graph associated to mobile agents moving along a line that is an one-dimensional space.
Consider $n$ mobile agents in an oriented horizontal line. 
At time $t_0=0$, they initially appear along the line according to an ordering $\pi_0$. 
These agents move in the line and can cross one another as time goes on. We assume that a crossing is always between exactly two neighboring agents, and a single pair of agents cross each other at a single time. By ordering the crossings, we have the $k$-th crossing happening at time $t_k=k$.

A \emph{1D-mobility schedule} from an ordering $\pi_0=a_1,\ldots,a_n$ of $n$ agents is a sequence $x=x_1,\ldots,x_T$ of crossings within the agents. 
Each crossing $x_t$ is a pair $uv$ indicating that agents $u$ and $v$ cross each other at time $t$. Note that their ordering $\pi_{t}$ at time $t$ is obtained from $\pi_{t-1}$ by exchanging $u$ and $v$, and it is thus required that they appear consecutively in $\pi_{t-1}$. To such a schedule, we can associate a temporal graph $\mathcal{G}_{\pi_0,x}=((V,E),\lambda)$ such that:
\begin{itemize}
    \item $V=\{a_1,\ldots,a_n\}$,
    \item 
    $E=\{uv : \exists t\in [T], x_t=uv\}$,
    \item 
    for all $uv\in E$, $\lambda(uv)=\{t : x_t=uv\}$.
\end{itemize}

 We are interested in particular by the case where all agents cross each other exactly once as the resulting temporal graph is then a temporal clique which
 is called a \emph{1D-mobility temporal clique}.
 More generally, we say that an incremental temporal graph $\mathcal{G}$ \emph{corresponds to a 1D-mobility schedule} if there exists some ordering $\pi$ of its vertices and a 1D-mobility schedule $x$ from $\pi$ such that the identity is an isomorphism from $\mathcal{G}$ to $\mathcal{G}_{\pi,x}$. It is then called a \emph{1D-mobility temporal graph}.

\subsection{Reduced decomposition of a permutation}

Our definition of mobility model is tightly related to the notion of reduced decomposition of a permutation~\cite{tenner2006reduced}. Let $\mathcal{S}_n$ denote the symmetric group on $n$ elements. We represent a permutation $w \in S_n$ as a sequence $w=w(1),\ldots,w(n)$ and define its length $l(w)$ as the number of inverse pairs in $w$, i.e.
$l(w)= |\{ i,j : i<j, w(i)>w(j)\}|$. 
A \emph{sub-sequence} $w'$ of $w$ is defined by its length $k\in [n]$ and indices $1\le i_1<\cdots<i_k\le n$ such that $w'=w(i_1),\ldots,w(i_k)$.

A \emph{transposition} $\tau=(i,j)$ is the transposition of $i$ and $j$, that is $\tau(i)=j$, $\tau(j)=i$ and $\tau(k)=k$ for $k\in [n]\setminus\{i,j\}$. It is an \emph{adjacent} transposition when $j=i+1$.
Given a permutation $w$ and an adjacent transposition $\tau=(i,i+1)$, we define the \emph{right product} of $w$ by $\tau$ as the composition $w\tau=w\circ \tau$. Note that $w'=w\tau$, as a sequence, is obtained from $w$ by exchanging the numbers in positions $i$ and $i+1$ as we have $w'(i)=w(\tau(i))=w(i+1)$, $w'(i+1)=w(\tau(i+1))=w(i)$ and $w'(k)=w(k)$ for $k\not= i,j$.
A \emph{reduced decomposition} of a permutation $w \in S_n$ with length $l(w)=l$, is a sequence of adjacent transpositions $\tau_1,\tau_2,\ldots, \tau_l$ such that we have $w = \tau_{1}\ldots \tau_{l}$. 
Counting the number of reduced decompositions of a permutation has been well studied (see in particular \cite{stanley1984number}).

The link with our 1D-mobility model is the following. Consider a 1D-mobility schedule $x$ from an ordering $\pi_0$. Without loss of generality we assume that agents are numbered from $1$ to $n$. Each ordering $\pi_t$ is then a permutation. If  agents $u$ and $v$ cross at time $t$, i.e. $x_t=uv$, and their positions in $\pi_{t-1}$ are $i$ and $i+1$, we then have $\pi_{t}=\pi_{t-1}\tau_t$ where $\tau_t=(i,i+1)$. If each pair of agents cross at most once, then one can easily see that the schedule $x$ of crossings corresponds to a reduced decomposition $\tau_1,\ldots,\tau_T$ of $\pi_0^{-1}\pi_T=\tau_1\cdots\tau_T$ as the ending permutation is $\pi_T=\pi_0\tau_1\cdots\tau_T$. Note that this does not hold if two agents can cross each other more than once as the length of the schedule can then be longer than the length of $\pi_0^{-1}\pi_T$.

Interestingly, another decomposition is obtained by interpreting the crossing $x_t=uv$ at time $t$ as the transposition $(u,v)$. We then have $\pi_t=x_t \pi_{t-1}$ for each time $t$, and finally obtain $\pi_T=x_T\cdots x_1\pi_0$. Note that given an arbitrary sequence of transpositions $x_1,\ldots,x_T$, it is not clear how to decide whether there exists an ordering $\pi_0$ and a corresponding sequence of \emph{adjacent} transpositions $\tau_1,\ldots,\tau_T$ such that $x_t\cdots x_1\pi_0=\pi_0\tau_1\cdots\tau_t$ for all $t\in [T]$. This is basically the problem we address in the next section.

\section{1D-mobility temporal cliques}\label{sec:charac}

\subsection{Characterization}

Consider the ordered temporal patterns from Figure~\ref{fig:patterns} with respect to the initial ordering of the nodes in a 1D-mobility schedule $x$ producing a temporal clique $\mathcal{G}_x$. One can easily see that the upper-left pattern cannot occur in $\mathcal{G}_x$ within three agents $a,b,c$ appearing in that order initially: $a$ and $c$ cannot cross each other as long as $b$ is still in-between them, while the pattern requires that edge $ac$ appears before $ab$ and $bc$. A similar reasoning prevents the presence of the three other patterns. It appears that excluding these four patterns suffices to characterize 1D-mobility temporal cliques, as stated bellow. 

\begin{theorem}\label{th:patterns}
    A temporal clique  is a 1D-mobility temporal clique if and only if there exists an ordering of its nodes that excludes the four ordered temporal patterns of Figure~\ref{fig:patterns}.
\end{theorem}

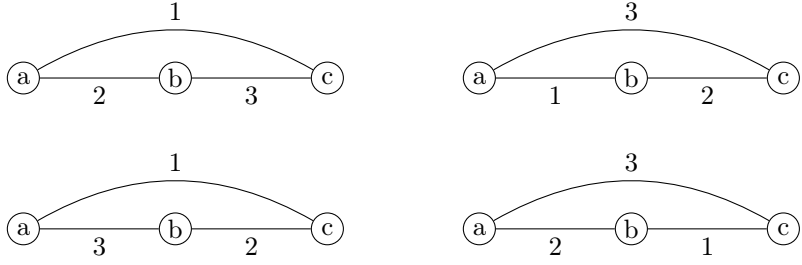
\begin{figure}[htp]
\begin{center}
\begin{tikzpicture}
   \tikzstyle{circlenode}=[draw,circle,minimum size=70pt,inner sep=0pt]
    \tikzstyle{whitenode}=[draw,circle,fill=white,minimum size=12pt,inner sep=0pt]
 
\draw (0,0) node[whitenode] (a1)   [] {a};
\draw (2,0) node[whitenode] (b1)   [] {b};
\draw (4,0) node[whitenode] (c1)   [] {c};
\draw (a1) edge node [below] {2} (b1);
\draw (b1) edge node [below] {3} (c1);
\draw (a1) [bend left=30] edge node [above] {1} (c1);

\draw (6,0) node[whitenode] (a2)   [] {a};
\draw (8,0) node[whitenode] (b2)   [] {b};
\draw (10,0) node[whitenode] (c2)   [] {c};
\draw (a2) edge node [below] {1} (b2);
\draw (b2) edge node [below] {2} (c2);
\draw (a2) [bend left=30] edge node [above] {3} (c2);

\draw (0,-2) node[whitenode] (a3)   [] {a};
\draw (2,-2) node[whitenode] (b3)   [] {b};
\draw (4,-2) node[whitenode] (c3)   [] {c};
\draw (a3) edge node [below] {3} (b3);
\draw (b3) edge node [below] {2} (c3);
\draw (a3) [bend left=30] edge node [above] {1} (c3);

\draw (6,-2) node[whitenode] (a4)   [] {a};
\draw (8,-2) node[whitenode] (b4)   [] {b};
\draw (10,-2) node[whitenode] (c4)   [] {c};
\draw (a4) edge node [below] {2} (b4);
\draw (b4) edge node [below] {1} (c4);
\draw (a4) [bend left=30] edge node [above] {3} (c4);
\end{tikzpicture}
    
\end{center}
\caption{Ordered forbidden patterns in an ordered 1D-mobility temporal clique. Each pattern is ordered from left to right and has associated ordering $a,b,c$.}\label{fig:patterns}
\end{figure}

Let $\mathcal{C}$ denote the class of temporal cliques which have an ordering excluding the four ordered temporal patterns of Figure~\ref{fig:patterns}.
We first prove that any 1D-mobility temporal clique is in $\mathcal{C}$:

\begin{proposition}\label{prop:1DtoPat}
For any 1D-mobility schedule $x$ from an ordering $\pi$ of $n$ agents such that $\mathcal{G}_{\pi,x}=((V,E),\lambda)$ is a temporal clique, the initial ordering $\pi$ excludes the four patterns of Figure~\ref{fig:patterns}.
\end{proposition}

This proposition is a direct consequence of the following lemma.

\begin{lemma}\label{lem:triangle}
    Consider three nodes $a,b,c\in V$ such that time $\lambda(ac)$ happens  in-between $\lambda(ab)$ and $\lambda(bc)$, i.e. $\lambda(ac)$ is the median of $\set{\lambda(ab),\lambda(ac), \lambda(bc)}$, then $b$ is in-between $a$ and $c$ in the initial ordering, i.e. either $a,b,c$ or $c,b,a$ is a sub-sequence of $\pi$.
\end{lemma}

\begin{proof}
  For the sake of contradiction, suppose that $b$ is not in-between $a$ and $c$ initially. At time $\min\set{\lambda(ab),\lambda(bc)}$, it first crosses $a$ or $c$, and it is now in-between $a$ and $c$. 
  As $a$ and $c$ cannot cross each other as long as $b$ lies in-between them, the other crossing of $b$ with $a$ or $c$ should thus occur before $\lambda(ac)$, implying $\max\set{\lambda(ab),\lambda(bc)}<\lambda(ac)$, in contradiction with the hypothesis. 
  The only possible initial orderings of these three nodes are thus $a,b,c$ and $c,b,a$.
\end{proof}

One can easily check that the above Lemma \ref{lem:triangle} forbids the four patterns of Figure~\ref{fig:patterns}. Indeed, in each pattern, the edge of label 2 that appears in-between the two others in time, is adjacent to the middle node while it should link the leftmost and rightmost nodes. Proposition~\ref{prop:1DtoPat} thus follows.

\smallskip
We now show that forbidding these four patterns fully characterizes 1D-mobility temporal cliques.
For that purpose, we construct a mapping from ordered temporal cliques in $\mathcal{C}$ to the set $R(w_n)$ of all reduced decompositions of $w_n$ where
 $w_n=n,n-1,\ldots,1$ is the permutation in $\mathcal{S}_n$  with  longest length.

\begin{lemma}{\label{lem: mapping}}
Any temporal clique $\mathcal{G} \in  \mathcal{C}$ having an ordering $\pi$ excluding the four patterns of Figure~\ref{fig:patterns}, can be associated  to a reduced decomposition $f(\mathcal{G},\pi)\in R(w_n)$ of $w_n$.  Moreover, the representation $\mathcal{R}(\mathcal{G})$ of $\mathcal{G}$ corresponds to a 1D-mobility schedule starting from $\pi$ and $\mathcal{G}$ is a 1D-mobility temporal clique.
\end{lemma}

\begin{proof}
Recall that, up to isomorphism, we can assume that $\mathcal{G}$ has lifetime $T=n(n-1)/2$ and that exactly one edge appears at each time $t\in [T]$.
Consider the corresponding representation $\mathcal{R}(\mathcal{G})=(u_1,v_1,1),(u_2,v_2,2)$, $\ldots,(u_T,v_T,T)$.
Starting from the initial ordering $\pi_0=\pi$, we construct a sequence $\pi_1,\ldots,\pi_T$ of orderings
corresponding to what we believe to be the positions of the agents at each time step if we read the edges in $\mathcal{R}(\mathcal{G})$ as a 1D-mobility schedule. More precisely, for each $t\in T$,  $\pi_t$ is defined from $\pi_{t-1}$ as follows. As the edge $u_tv_t$ should correspond to a crossing $x_t=u_tv_t$, it can be seen as the transposition exchanging $u_t$ and $v_t$ so that we define $\pi_t=x_t\pi_{t-1}$. Equivalently, we set $\tau_t=(i,j)$ where $i$ and $j$ respectively denote the indexes of $u_t$ and $v_t$ in $\pi_{t-1}$, i.e. $\pi_{t-1}(i)=u_t$ and $\pi_{t-1}(j)=v_t$. We then also have $\pi_t=\pi_{t-1}\tau_t$.

Our main goal is to prove that $f(\mathcal{G},\pi):=\tau_1,\ldots,\tau_T$ is the desired reduced decomposition of $w_n$.
For that, we need to prove that $u_t$ and $v_t$ are indeed adjacent in $\pi_{t-1}=\pi_0\tau_1\cdots\tau_{t-1}=x_{t-1}\cdots x_1\pi_0$. For the sake of contradiction, consider the first time $t$ when this fails to be. That is $\tau_1,\ldots,\tau_{t-1}$ are indeed adjacent transpositions, edge $uv$ appears at time $t$, i.e. $uv=u_tv_t$, and $u,v$ are not consecutive in $\pi_{t-1}$.
We assume without loss of generality that $u$ is before $v$ in $\pi_0$, i.e. $u,v$ is a sub-sequence of $\pi_0$.
We will mainly rely on the following observation:

Consider two nodes $x,y$ such that $x$ is before $y$ in $\pi_0$, then $x$ is before $y$ in $\pi_{t-1}$ if and only edge $xy$ appears at $t$ or later, i.e. $\lambda(xy)\ge t$.

The reason comes from the assumption that $\tau_1,\ldots,\tau_{t-1}$ are all adjacent transpositions: as long as only $x$ or $y$ is involved in such a transposition, their relative order cannot change.
The above observation thus implies in particular that $u$ is still before $v$ in $\pi_{t-1}$.
Now, as $u$ and $v$ are not consecutive in $\pi_{t-1}$, 
there must exist an element $w$ between elements $u$ and $v$ in $\pi_{t-1}$. We consider the two following cases:

Case 1. 
$w$ was already in-between $u$ and $v$ in $\pi_0$, that is $u,w,v$ is a sub-sequence of $\pi_0$.
As the relative order has not changed between these three nodes, we have $\lambda(uw)> t$ and $\lambda(wv)> t$ as their appearing time is distinct from   $t=\lambda(uv)$. This is in contradiction with the exclusion of the two patterns on the left of Figure~\ref{fig:patterns}.

Case 2. 
$w$ was not in-between $u$ and $v$ in $\pi_0$. Consider the case where $u,v,w$ is a sub-sequence of $\pi_0$. From the observation, we we deduce that $\lambda(vw)< t$ and $\lambda(uw)> t$, which contradicts the exclusion of the
bottom-right pattern of Figure~\ref{fig:patterns}.
The other case where $w,u,v$ is a sub-sequence of $\pi_0$ is symmetrical and similarly leads to a contradiction with the exclusion of the top-right pattern of Figure~\ref{fig:patterns}.

We get a contradiction in all cases and conclude that
$\tau_1,\ldots,\tau_T$ are all adjacent transpositions. This implies that $x$ is indeed a valid 1D-mobility schedule from $\pi$. As $x$ is defined according to the ordering of edges in $\mathcal{R}(\mathcal{G})$ by appearing time,  $\mathcal{G}$ is obviously isomorphic to $\mathcal{G}_{\pi,x}$.

Additionally, as each pair of elements occurs exactly in one transposition, the permutation $\tau_1\cdots\tau_T$ has length $T=n(n-1)/2$ and must equal $w_n$. The decomposition $f(\mathcal{G},\pi)=\tau_1,\ldots,\tau_T$ is thus indeed a reduced decomposition of $w_n$.
\end{proof}
\vspace{0.3cm}

Theorem~\ref{th:patterns} is a direct consequence of Proposition~\ref{prop:1DtoPat} and Lemma~\ref{lem: mapping}.

\subsection{Recognition algorithm}

 We now propose an Algorithm to  decide if a clique belongs to $\mathcal{C}$, and provide an ordering of the nodes that avoids the patterns if it is the case.   
The main idea of the algorithm relies on Lemma~\ref{lem:triangle} which allows us to detect within a triangle which node should be in-between the two others in any ordering avoiding the patterns. This is simply done by comparing the three times at which the edges of the triangle appear.
We assume the input to be given as a representation of the temporal graph, i.e. the list $\mathcal{R}(\mathcal{{G}})$ of the edges in the form $(u,v,t)$, sorted according to their time labels. 

The algorithm first tries to compute a feasible ordering of the vertices using the function $\vertexsorting(\mathcal{R}(\mathcal{{G}}))$. To do that, the subroutine $\extremalnodes(V)$ provides the two nodes that should be at the extremities of some subset $V$ of nodes. It outputs these two nodes by excluding repeatedly a node out of some triplets again and again until only two nodes are left, using Lemma~\ref{lem:triangle} to identify which one is in the middle. 

\begin{algorithm}

\Function{$\vertexsorting(\mathcal{R}(\mathcal{G}))$}{
\Input{The representation $\mathcal{R}(\mathcal{G})$ of a temporal clique $\mathcal{G}=((V,E),\lambda)$.}
\Output{A vertex ordering $\pi$.}
Compute a matrix representing $\lambda$ and the set $V$ of vertices from $\mathcal{R}(\mathcal{G})$.\;
$X:=\extremalnodes(V)$\;
Let $a$ and $z$ be the two vertices in $X$.\\ 
Define a partial ordering $\pi$ with $a$ as single element.\label{lin:5}\\
$V:=V\setminus\set{a}$\;
\While{$V\not=\set{z}$\label{lin:7}}{
  $X:=\extremalnodes(V)$\label{lin:8}\;
  \eIf{$z\notin X$}{
    \return{$\bottom$}\quad\Comment{Failure.}\label{line:failure}
  }{ 
  Let $b$ be the node in $X\setminus\set{z}$.\\
  Append $b$ to $\pi$.\\
  $V:=V\setminus\set{b}$\label{lin:14}\;
  }
}
Append $z$ to $\pi$.\\
\Return{$\pi$}
}

\smallskip

\Function{$\extremalnodes(V)$}{
  Let $W$ be a copy of $V$.\\
  Pick any pair $u,v$ of nodes in $W$.\label{lin:19}\;
  $W:=W\setminus\set{u,v}$\;
  Set $X:=\set{u,v}$. \Comment{Tentative pair of extremal nodes.}
  \While{$W\not=\emptyset$\label{lin:22}}{
    Remove a node $w$ from $W$.\label{lin:23}\;
    $X:=\triangleextremities(X\cup\set{w})$\label{lin:24}\;
  }
  \Return{$X$}
}
\smallskip
\Function{$\triangleextremities(T)$}{
  Let $u,v,w$ be the three nodes in $T$.\\
  Retrieve the three time labels $\lambda(uv), \lambda(vw),\lambda(uw)$.\\
  \Return{the pair consisting of the edge $e\in\set{uv, vw, uw}$ with median time label.}
}

\end{algorithm}

We thus finally get the two extremities $a$ and $z$ of $V$ in the initial ordering. Without loss of generality, we keep $a$ as the first element. We then repeat $n-2$ times the following procedure: add back $z$ to the remaining nodes, and compute the two extremities among them. If $z$ is one of the extremities, remove the other one and append it to the partial ordering constructed so far. Otherwise, return $\bottom$ as a contradiction has been found ($z$ must always be an extremity if we have a 1D-mobility temporal clique).

We then need  to check that each temporal edge indeed exchanges two consecutive nodes in the 1D-mobility model. To do that, we represent the sequence of permutations starting from $\pi$ the initial ordering, and check that each switch, according to the edges sorted by time label, corresponds to an exchange between two consecutive nodes. If at some point, we try to switch non consecutive elements, we return $False$, otherwise at the end we have checked that we had a 1D-mobility temporal clique and return $True$.

\begin{algorithm}
\caption{The recognition algorithm.}\label{alg:vertex-sorting}

\Input{A temporal clique $\mathcal{G}=((V,E),\lambda)$ given by its representation $\mathcal{R}(\mathcal{G})$.}
\Output{True if $\mathcal{G}$ excludes the four forbidden patterns of Figure~\ref{fig:patterns}, False otherwise.}

   $\pi:=\vertexsorting(\mathcal{R}(\mathcal{G}))$\;
   \lIf{$\pi=\bottom$}{\return{False}}
   Compute the index $\sigma(v)$ of each vertex $v$ in $\pi$.\\
    \For{each triplet $(u,v,t)$ in $\mathcal{R}(\mathcal{G})$ scanned in increasing order\label{lin:begfor}}{
        \If{$|\sigma(u)-\sigma(v)|=1$
        }{ \Comment{$u$ and $v$ are consecutive in $\pi$.}
            Exchange $u$ and $v$ in $\pi$ and update $\sigma$.\\
        }
        \Else {\Return False}\label{lin:endfor}
    }
    \Return True

\end{algorithm}

\begin{theorem}\label{prop:vertex-sorting}
    Algorithm~\ref{alg:vertex-sorting} correctly recognizes 1D-mobility temporal cliques in linear time and outputs an initial vertex ordering from which its representation is a 1D-mobility schedule.
\end{theorem}

\begin{proof}
Let $\pi'$ be an ordering that excludes the four forbidden patterns from $\mathcal{G}$.

    For any $W\subseteq V$, the extreme vertices of $W$ returned by $\extremalnodes(W)$ are uniquely defined by $W$. In other words, these extreme vertices do not depend on the order of nodes $w$ picked from $W$ in Line~\ref{lin:23}, and two nodes, $u$ and $v$, picked from $W$ in Line~\ref{lin:19}.
    The reason is that, in the loop from Line~\ref{lin:22} to Line~\ref{lin:24}, the node in-between the two others in the restriction of $\pi'$ to the three nodes in $X\cup\{w\}$ will be eliminated by Lemma~\ref{lem:triangle}. The function $\extremalnodes(W)$ thus returns the two extreme vertices of the restriction of $\pi'$ to $W$. 
        
    We now prove that the ordering $\pi$ returned by function $\vertexsorting(\mathcal{R}(\mathcal{G}))$ is either $\pi'$ itself or the reverse of $\pi'$. Furthermore, the ordering $\pi$ is uniquely defined, depending on the choice of the first element $a$ in Line~\ref{lin:5} (among the the two elements in $X$). 
    At the beginning of the $i$-th iteration of the While loop from Line~\ref{lin:7} to Line~\ref{lin:14}, the ordering $\pi$ includes $i$ elements, denoted by $x_1, x_2,\ldots, x_{i}$ where $x_1 = a$. 
    We prove by induction on $i$ that function $\vertexsorting$ eventually returns an ordering $\pi$ that is either $\pi'$ or the reverse of $\pi'$. For $i=1$, $x_1 = a$ is an extreme vertex in $\pi'$.
    Assume that  $x_{j}$ and $x_{j+1}$ appear consecutively in the ordering $\pi'$, i.e., $|(\pi')^{-1}(x_j)-(\pi')^{-1}(x_{j+1})| =1$, for all $j\leq i$. 
    In Line~\ref{lin:8}, set $X$ includes two extreme vertices among all vertices that do not appear in $\pi$, according to $\pi'$. The two vertices in set $X$ are thus uniquely defined, and include $z$ and an other vertex, call it $b$, which is in-between $a$ and $z$ in $\pi'$. 
    Additionally, since $x_{j}$ and $x_{j+1}$ appear consecutively in the ordering $\pi'$, for all $j\leq i$, it implies that $b$ is adjacent to $x_i$ in $\pi'$. 
    At the end of this iteration, $x_{i+1}=b$ is appended to the tuple $\pi$. 

    Now we analyse the running time. The subroutine $\extremalnodes(W)$ clearly runs in $O(|W|)$ time as $\triangleextremities(T)$ takes constant time. $\vertexsorting(\mathcal{R}(\mathcal{G}))$ thus uses $O(n^2)$ time which is linear since a temporal clique has $\frac{n(n-1)}{2}$ edges (this can be checked beforehand). Finally, checking that $\mathcal{R}(\mathcal{G})$ is indeed a 1D-mobility schedule from $\pi$ in the for loop at Lines~\ref{lin:begfor}-\ref{lin:endfor} clearly takes constant time per triplet.
\end{proof}

\subsection{Counting}

We now estimate the number $\card{\mathcal{C}}$ of 1D-mobility temporal cliques with $n$ nodes through the following result. Recall that $R(w_n)$ denotes the set  of all reduced decompositions of $w_n$ where
 $w_n=n,n-1,\ldots,1$.

\begin{proposition}\label{prop:counting}
    The number of 1D-mobility temporal cliques with $n$ nodes is 
    $$
    \card{\mathcal{C}}=\frac{\card{R(w_n)}}{2}
    = \frac{1}{2} \frac{{n \choose 2}!}{1^{n-1}3^{n-2} \cdots (2n-3)^1}.
    $$
\end{proposition}

Let us define $\mathcal{C'}\subseteq \mathcal{C}\times \mathcal{S}_n$ as the set of ordered temporal cliques $(\mathcal{G},\pi)$ such that $\mathcal{G}=((V,E),\lambda)\in \mathcal{C}$ and $\pi$ is an ordering of $V$ such that $\mathcal{R}(\mathcal{G})$ provides a 1D-mobility schedule from $\pi$.
Proposition~\ref{prop:counting} derives from two following lemmas and known results~\cite{stanley1984number} counting the number $\card{R(w_n)}$ of reduced decompositions of $w_n$ according to the hook length formula~\cite{frame1954hook}.

\begin{lemma}
    The mapping $f:\mathcal{C'}\to R(w_n)$ defined in Lemma~\ref{lem: mapping} is a bijection.
\end{lemma}
\begin{proof}
  We simply define a mapping $g:R(w_n)\to\mathcal{C'}$ such that $f\circ g$ is the identity. Consider a reduced decomposition $\rho=\tau_1,\ldots,\tau_T$ of $w_n$ where each $\tau_t$ is an adjacent transposition. As $w_n$ has $n(n-1)/2$ inversions, its length is indeed $T=n(n-1)/2$. Let $\pi_0=1,\ldots,n$ be the identity permutation and define $\pi_t=\pi_0\tau_1\cdots\tau_t=\tau_1\cdots\tau_t$ for each $t\in [T]$. Let $x_t=uv$ be the pair of elements in position $i$ and $i+1$ in $\pi_{t-1}$ where $i$ is the index such that $\tau_t=(i,i+1)$. We then have $\pi_t=x_t \pi_{t-1}$ for all $t\in [T]$, and $x$ is indeed a 1D-mobility schedule from $\pi_0$ that leads to $w_n$.  As each pair of agents $u,v\in \pi_0$ with $u<v$ appears as sub-sequence $u,v$ in $\pi_0$ and sub-sequence $v,u$ in $\pi_T=w_n$, they must cross at some time $t$ such that $x_t=uv$. As the total number of crossings is $T=n(n-1)/2$, this can happen only once, and $\mathcal{G}_{\pi_0,x}$ is a temporal clique. We can thus define $g(\rho)=(\mathcal{G}_{\pi_0,x},\pi_0)$ 
  which satisfies
  $f(g(\rho))=\rho$ as we have $x_t\cdots x_1\pi_0=\pi_t=\pi_0\tau_1\cdots\tau_t$ for all $t\in [T]$. 
\end{proof}

\begin{lemma}
    Any 1D-mobility temporal clique $\mathcal{G}$ admits exactly two orderings $\pi$ and its reversal such that $\mathcal{R}(\mathcal{G})$ provides a 1D-mobility schedule from $\pi$.
\end{lemma}
\begin{proof}
    Let $\pi$ be the ordering provided by Algorithm~\ref{alg:vertex-sorting} on a given 1D-mobility clique. By Theorem~\ref{prop:vertex-sorting}, we know that $\pi$ is an ordering corresponding to this clique. We can notice that $rev(\pi)$, the reversed order of $\pi$, is also an ordering corresponding to that clique.
     
    Note that, by Lemma \ref{lem:triangle}, for any three nodes, $u$, $v$, $w$, of $V$, with edge label $\lambda(uw)<\lambda(uv)<\lambda(wv)$, $w$ has to be in-between $u$ and $v$ in all orderings that excludes the forbidden patterns from $\mathcal{G}$. This implies that no other ordering than $\pi$ and $rev(\pi)$ can correspond to this clique.

    It implies that any 1D-mobility temporal clique admits exactly two orderings such that $\mathcal{R}(\mathcal{G})$ provides a 1D-mobility schedule from $\pi$.
\end{proof}

\subsection{Temporal spanner}

In this subsection, we show that any 1D-mobility temporal clique has a spanner of $\mathcal{G}$ of size $(2n-3)$. Moreover, this spanner has diameter 3, and provides a new structure for (sub)spanners compared to the ones introduced in~\cite{casteigts2021temporal}, see Figure~\ref{fig:spanner2} below.

\begin{theorem}
Given a 1D-mobility temporal clique $\mathcal{G}$, let $\mathcal{H}$ be the temporal subgraph of $\mathcal{G}$ consisting in the $(2n-3)$ edges that are adjacent with either $v_1$ or vertex $v_n$. $\mathcal{H}$ is a temporal spanner of $\mathcal{G}$.
\end{theorem}

\begin{proof}
Let us consider the edge $(v_1,v_n,t)$ that corresponds to the crossing of the initial two extremities $v_1$ and $v_n$ on the line. When this happens, we have two sets: $V_L$ (resp. $V_R$) corresponding to the agents being at the left of $v_1$ (resp. right of $v_n$) at time $t$. All edges of $\cal H$ appearing before $t$, have the form $v_1v_l$ with $v_l\in V_L$ or $v_rv_n$ with $v_r\in V_R$. All edges of $\cal H$ appearing after $t$ have the form $v_1v_r$ with $v_r\in V_R$ or$v_lv_n$ with $v_l\in V_L$ (see Figure \ref{fig:spanner2}).

\begin{figure}[htp]
\begin{center}
\begin{tikzpicture}
   \tikzstyle{circlenode}=[draw,circle,minimum size=70pt,inner sep=0pt]
    \tikzstyle{whitenode}=[draw,circle,fill=white,minimum size=12pt,inner sep=0pt]
    \tikzstyle{bignode}=[draw,circle,fill=white,minimum size=40pt,inner sep=0pt]
 
\draw (1.5,0) node[bignode] (l)   [] {$V_L$};
\draw (8.5,0) node[bignode] (r)   [] {$V_R$};
\draw (5,0.75) node[whitenode] (v1)   [] {$v_1$};
\draw (5,-0.75) node[whitenode] (vn)   [] {$v_n$};
\draw (l) edge [->] node [above] {$t'<t$} (v1);
\draw (l) edge [<-] node [below] {$t'>t$} (vn);
\draw (r) edge [->] node [below] {$t'<t$} (vn);
\draw (r) edge [<-] node [above] {$t'>t$} (v1);
\draw (v1) edge [<->] node [right] {$t$} (vn);

\end{tikzpicture}
\end{center}
\caption{Relative order of labels of edges between the sets $V_L$, $V_R$ and the two vertices $v_1$ and $v_n$, showing how temporal paths between $V_L$ and $V_R$ can be formed.}\label{fig:spanner2}
\end{figure}
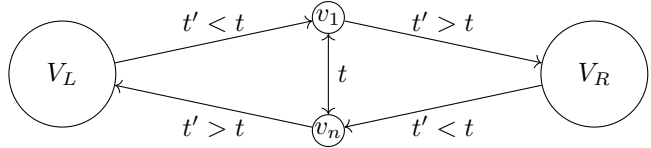

As we have all edges connected to $v_1$ and $v_n$, we only need to prove that we keep connectivity between $V_L$ and $V_R$, but also within those sets:
\begin{itemize}
    \item To connect a node $v_l\in V_L$ to $v_r\in V_R$, we use the path $(v_l,v_1,v_r)$.
    \item To connect a node $v_r\in V_R$ to $v_l\in V_L$, we use the path $(v_r,v_n,v_l)$.
    \item To connect a node $v_l\in V_L$ to $v_l'\in V_L$, we use the path $(v_l,v_1,v_n,v_l')$.
    \item To connect a node $v_r\in V_R$ to $v_r'\in V_R$, we use the path $(v_r,v_n,v_1,v_r')$.
\end{itemize}

\end{proof}

\subsection{Agents moving with constant speed}
Recall that any 1D-mobility schedule corresponds to a set of agents moving continuously on an oriented horizontal line with (possibly) varying speeds.
When studying temporal cliques generated by 1D-mobility schedules,  it is natural to ask whether they can be realised using the following simpler, constant speed model:
\begin{itemize}
    \item  agent $i$ starts in position $p_i \in \mathbb R$ with an initial ordering $p_1 < p_2 < \dots < p_n$
    \item agent $i$ moves towards infinity with speed $s_i$ that satisfies $s_1 > s_2 > \dots > s_n$.
\end{itemize}
These assumptions guarantee that each agent can meet at most once, and that after a certain time, the order of agents reverses thus we obtain a temporal clique. (By slightly perturbing the starting positions, we can assume that at most one pair of agents meet at any moment.)
In this section, we show that this simplification comes with a loss of generality.
\begin{theorem}
There exists a 1D-mobility schedule of $7$ agents that cannot be realized as a temporal graph associated to mobile agents moving with constant speeds.
\end{theorem}
\begin{proof}
Denote the agents by $a,b,c,d,e,f,g$, starting in this order, and consider the following sequence of crossings:
\begin{align}\label{eq:speed-crossing-seq}
\begin{split}
abcdefg 
&~ \overset{x_{1}=ab}{\xrightarrow{\hspace*{.6cm}}} ~ bacdefg 
~\overset{x_{2}=cd}{\xrightarrow{\hspace*{.6cm}}}~ badcefg 
~\overset{x_{3}=ad}{\xrightarrow{\hspace*{.6cm}}}~ bdacefg
~\overset{x_{4}=ef}{\xrightarrow{\hspace*{.6cm}}} ~bdacfeg 
~\overset{x_{5}=eg}{\xrightarrow{\hspace*{.6cm}}} ~bdacfge\\
&~\overset{x_{6}=ac}{\xrightarrow{\hspace*{.6cm}}}~ bdcafge
~\overset{x_{7}=af}{\xrightarrow{\hspace*{.6cm}}} ~bdcfage
~\overset{x_{8}=bd}{\xrightarrow{\hspace*{.6cm}}} ~dbcfage
~\overset{x_{9}=bc}{\xrightarrow{\hspace*{.6cm}}} ~dcbfage
~\overset{x_{10}=ag}{\xrightarrow{\hspace*{.6cm}}}~ dcbfgae.
\end{split}
\end{align}
The corresponding temporal graph is illustrated in Figure~\ref{fig:speedpattern} with nodes ordered from left to right according to their initial ordering. 
\begin{figure}[htp]
\begin{center}
\begin{tikzpicture}
   \tikzstyle{circlenode}=[draw,circle,minimum size=70pt,inner sep=0pt]
    \tikzstyle{whitenode}=[draw,circle,fill=white,minimum size=12pt,inner sep=0pt]
 
\draw (0,0) node[whitenode] (a)   [] {a};
\draw (2,0) node[whitenode] (b)   [] {b};
\draw (4,0) node[whitenode] (c)   [] {c};
\draw (6,0) node[whitenode] (d)   [] {d};
\draw (8,0) node[whitenode] (e)   [] {e};
\draw (10,0) node[whitenode] (f)   [] {f};
\draw (12,0) node[whitenode] (g)   [] {g};
\draw (a) edge node [below] {\small 1} (b);
\draw (b) edge node [below] {\small 9} (c);
\draw (c) edge node [below] {\small 2} (d);
\draw (e) edge node [below] {\small 4} (f);
\draw (a) [bend left=45] edge node [below] {\small 10} (g);
\draw (a) [bend left=40] edge node [below] {\small 6} (c);
\draw (a) [bend left=45] edge node [below] {\small 3} (d);
\draw (e) [bend left=45] edge node [below] {\small 5} (g);
\draw (a) [bend right=30] edge node [below] {\small 7} (f);
\draw (b) [bend right=40] edge node [below] {\small 8} (d);
\end{tikzpicture}
    
\end{center}
\caption{Any 1D-mobility clique with the above prefix cannot be realized with agents moving at constant speed.}\label{fig:speedpattern}
\end{figure}

\noindent
It is easy to verify that Sequence~\eqref{eq:speed-crossing-seq} can be extended to a 1D-mobility temporal clique in the following way:
\begin{align*}
\begin{split}
dcbfgae 
&~ \overset{x_{11}=ae}{\xrightarrow{\hspace*{.6cm}}} ~ dcbfgea 
~\overset{x_{12}=bf,~ x_{13}=bg,~ x_{14}=be}{ \xrightarrow{\hspace*{3cm}}}~ dcfgeba 
~\overset{x_{15}=cf,~ x_{16}=cg,~ x_{17}=ce}{ \xrightarrow{\hspace*{3cm}}}~ dfgecba\\ 
&~\overset{x_{18}=df,~ x_{19}=dg,~ x_{20}=de}{ \xrightarrow{\hspace*{3cm}}}~ fgedcba
~\overset{x_{21}=fg}{ \xrightarrow{\hspace*{.6cm}}}~ gfedcba.
\end{split}
\end{align*}
\noindent Now assume that Sequence~\eqref{eq:speed-crossing-seq} can be obtained by agents moving with constant speeds $s_a > s_b> \dots > s_g$. 
The proof of impossibility relies on the following $4$ observations:
\begin{enumerate}[label=(\alph*)]
    \item The crossing subsequence $x_3 = ad, ~  x_4 = ef, ~  x_5 = eg, ~ x_6 = ac$ implies that the agent $a$ traverses the interval between  $d$ and $c$ \emph{slower} than agent $e$ traversing the interval between  $f$ and $g$. \label{obs:1}
    \item The crossing subsequence $x_7 = af, ~  x_8 = bd, ~  x_9 = bc, ~  x_{10} = ag$ implies that the agent $a$ traverses the interval between  $f$ and $g$ \emph{slower} than agent $b$ traversing the interval between  $d$ and $c$. \label{obs:2}
    \item Agents $f$ and $g$ do not cross for $t \in [0,10]$, thus in this time interval, their distance is decreasing as $t$ increases. This together with the facts that $s_a > s_e$ and that $a$ traverses the interval between $f$ and $g$ \emph{later} than $e$ does, imply that $e$ must traverse the interval between $f$ and $g$ \emph{slower} than $a$ does. \label{obs:3}
    \item Since $s_c > s_d$, the distance between agents $c$ and $d$ is increasing after they cross, that is, after time $t =2$. \label{obs:4}
\end{enumerate}
For agents $i,j,k$, let $T_i(j,k)$ denote the time  agent $i$ spends between $j,k$, that is, $T_i(j,k) = \lambda(ik) - \lambda(ij)$ (if $i$ is never located in-between $j$ and $k$, we set $T_i(j,k) = \infty$). Putting together Observations \ref{obs:1}-\ref{obs:3}, we get
\begin{equation*}
    T_a(d,c) ~\overset{\ref{obs:1}}{>}~ T_e(f,g) ~\overset{\ref{obs:3}}{>}~ T_a(f,g) ~\overset{\ref{obs:2}}{>}~ T_b(d,c)
\end{equation*}
However, $b$ traverses the interval between $d$ and $c$ \emph{strictly after} $a$ had traversed it and thus by Observation \ref{obs:4}, $b$ has to travel a larger distance between $d$ and $c$ than $a$ does. This, together with $s_a > s_b$ implies that $T_a(d,c) < T_b(d,c)$, a contradiction.
\end{proof}

We note that any 1D-mobility schedule on up to $4$ agents can be obtained with constant speed agents. The remaining cases of $5$ and $6$ agents are left as  open questions. We also conjecture that there are infinitely many critical patterns for 1D-mobility cliques that cannot be realized with constant-speed agents. On the other hand a characterization of 1D-mobility  schedules that can be realized with constant speed agents should be interesting.
\section{Mobility graph with at most one crossing}
\label{sec:atmost}

In this section, we consider the case where each pair of agents cross each other at most once. We provide a characterization with the addition of the forbidden patterns of Figure~\ref{fig:non_edges} which includes non-edges corresponding of the non-crossings.

\begin{proposition}\label{prop:Non_edges_Pat}
For any mobility schedule $x$ of $n$ agents where each pair of agents cross at most once producing an incremental single valued temporal graph $\mathcal{G}_x=((V,E),\lambda)$, the initial ordering $\pi$ of the agents excludes the patterns of Figures~\ref{fig:patterns} and~\ref{fig:non_edges}.
\end{proposition}

\begin{proof}
    Lemma~\ref{lem:triangle} gives us the proof for the patterns of Figure~\ref{fig:patterns}. About the patterns of Figure~\ref{fig:non_edges}, we have the following observations. Pick three nodes $a,b,c$ such that $a,b,c$ is a sub-sequence of $\pi$ and there exists two agents among them which do not cross each other.
    If $b$ and $c$ do not cross each other, and $a$ crosses $b$ and $c$, then $a$ crosses $b$ before crossing agent $c$ (top left pattern). If $b$ and $c$ do not cross each other, and $a$ does not cross $b$, then $a$ does not cross $c$ (bottom pattern). 
    If $a$ and $b$ do not cross, similarly, $\pi$ excludes the top right pattern and the bottom pattern.
    If $a$ and $c$ do not cross each other, then $b$ cannot cross both $a$ and $c$ (two patterns in the second row).
\end{proof}

\begin{figure}
    \begin{center}
\begin{tikzpicture}
   \tikzstyle{circlenode}=[draw,circle,minimum size=70pt,inner sep=0pt]
    \tikzstyle{whitenode}=[draw,circle,fill=white,minimum size=12pt,inner sep=0pt]
 
\draw (0,0) node[whitenode] (a7)   [] {a};
\draw (2,0) node[whitenode] (b7)   [] {b};
\draw (4,0) node[whitenode] (c7)   [] {c};
\draw (a7) edge node [below] {1} (b7);
\draw (b7) edge node [below] {2} (c7);
\draw (a7) [bend left=30, dotted] edge node {} (c7);

\draw (6,0) node[whitenode] (a6)   [] {a};
\draw (8,0) node[whitenode] (b6)   [] {b};
\draw (10,0) node[whitenode] (c6)   [] {c};
\draw (a6) edge  node [below] {2} (b6);
\draw (b6) edge node [below] {1} (c6);
\draw (a6) [bend left=30, dotted] edge node {} (c6);
 
\draw (0,1.2) node[whitenode] (a1)   [] {a};
\draw (2,1.2) node[whitenode] (b1)   [] {b};
\draw (4,1.2) node[whitenode] (c1)   [] {c};
\draw (a1) edge node [below] {2} (b1);
\draw (b1) edge [dotted] node {} (c1);
\draw (a1) [bend left=30] edge node [above] {1} (c1);

\draw (6,1.2) node[whitenode] (a2)   [] {a};
\draw (8,1.2) node[whitenode] (b2)   [] {b};
\draw (10,1.2) node[whitenode] (c2)   [] {c};
\draw (a2) edge [dotted] node  {} (b2);
\draw (b2) edge node [below] {2} (c2);
\draw (a2) [bend left=30] edge node [above] {1} (c2);

\draw (3,-1.2) node[whitenode] (a5)   [] {a};
\draw (5,-1.2) node[whitenode] (b5)   [] {b};
\draw (7,-1.2) node[whitenode] (c5)   [] {c};
\draw (a5) edge [dotted] node {} (b5);
\draw (b5) edge [dotted] node {} (c5);
\draw (a5) [bend left=30] edge node [above] {1} (c5);
\end{tikzpicture}
    
\end{center}
    \caption{Ordered forbidden patterns with forbidden edges in the 1D-mobility model when each pair of agents cross at most once. The ordering associated to each pattern is $a,b,c$.}
    \label{fig:non_edges}
\end{figure}
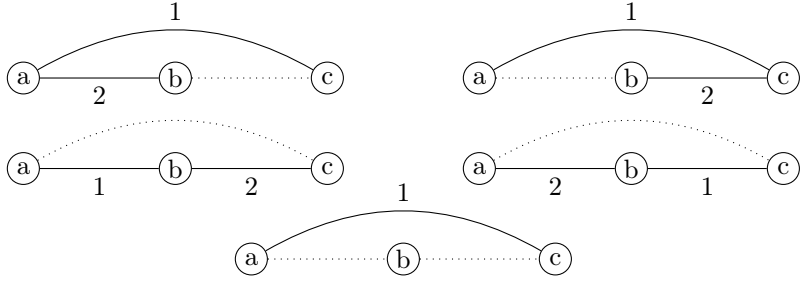

\begin{lemma}\label{lem:Non_edges_Pat}
Any incremental single valued temporal graph $\mathcal{G}$ having an ordering $\pi$ excluding the patterns of Figures~\ref{fig:patterns} and~\ref{fig:non_edges} can be associated  to a mobility schedule $x$ of $n$ agents where each pair of agents cross at most once. 
\end{lemma}
 
\begin{proof}
We proceed similarly as for the proof of Lemma~\ref{lem: mapping}.
Starting from the initial ordering $\pi_0=\pi$, we define a sequence $\pi_1,\ldots,\pi_T$, of orderings, corresponding to what, we believe, to be the positions of the agents at each time step if we read the edges of $\mathcal{G}$ by increasing time label as a mobility schedule. More precisely, for each $t\in T$, $\pi_t$ is defined from $\pi_{t-1}$ as follows. Consider the edge  $uv=\lambda^{-1}(t)$ appearing at time $t$ in $\mathcal{G}$. We define $\tau_t$ as the transposition exchanging $u$ and $v$ in $w_{t-1}$. Equivalently, we set $\tau_t=(i,j)$ where $i$ and $j$ respectively denote the indexes of $u$ and $v$ in $\pi_{t-1}$, i.e. $\pi_{t-1}(i)=u$ and $\pi_{t-1}(j)=v$. We then set $\pi_t=\pi_{t-1}\tau_t$.

We need to prove that $u_t$ and $v_t$ are indeed adjacent in $\pi_{t-1}$.
Consider the first time $t$ when this fails to be. That is $\tau_1, \ldots, \tau_{t-1}$ are adjacent transpositions, edge $uv$ appears at time $t$, i.e. $uv = u_tv_t$, and $u,v$ are not consecutive in $\pi_{t-1}$. It implies that there exists another node $w$ such that, in $\pi_{t-1}$, the relative ordering amongs the three vertices is $u_t, w, v_t$. We consider the temporal subgraph, $\mathcal{H}$, induced by the three vertices $u_t, w, v_t$. 
We get to a contradiction by proving that any order on those nodes in $\pi_0$ will imply that $\mathcal{H}$ is a forbidden pattern.

Case 1. $w$ was already in-between $u_t$ and $v_t$ in $\pi_0$. It implies that ($u_tw \notin E$ or $\lambda(u_tw)>t$) and ($v_tw \notin E$ or $\lambda(v_tw)>t$). This is forbidden by the two left patterns of Figure~\ref{fig:patterns}, the two top patterns and the bottom pattern of Figure~\ref{fig:non_edges}.

Case 2. $w$ was initially before $u_t$ and $v_t$ in $\pi_0$. It implies that $\lambda(u_tw)<t$, and $v_tw \notin E$ or $\lambda(v_tw)>t$. This is forbidden by the top right pattern of Figure~\ref{fig:patterns} and the third (middle left) pattern of Figure~\ref{fig:non_edges}.

Case 3. $w$ was initially after $u_t$ and $v_t$ in $\pi_0$. It implies that $\lambda(v_tw)<t$, and $u_tw \notin E$ or $\lambda(u_tw)>t$. This is forbidden by the bottom right pattern of Figure~\ref{fig:patterns} and the fourth (middle right) pattern of Figure~\ref{fig:non_edges}.
\end{proof}

Using Proposition \ref{prop:Non_edges_Pat} and Lemma \ref{lem:Non_edges_Pat}, we immediately obtain:

\begin{theorem}
    A single valued incremental temporal graph is a 1D-mobility temporal graph if and only if there exists an ordering of its nodes that excludes the ordered temporal patterns of Figures~\ref{fig:patterns} and~\ref{fig:non_edges}.
\end{theorem}

This characterization yields another argument to show that the graph in Figure \ref{fig:speedpattern} corresponds to 
1D temporal graph since it does not contain any of the patterns of Figures~\ref{fig:patterns} and  \ref{fig:non_edges}.

We make the following observation thanks to the characterization of graph classes through forbidden patterns of size 3 from~\cite{feuilloley2021graph}.

\begin{proposition}
The set of graphs that can be associated  to a mobility schedule $x$ of $n$ agents where each pair of agents cross at most once is contained in the set of permutation graphs.
    
\end{proposition}

\begin{proof}
First we can note that in the 5 patterns of Figure~\ref{fig:non_edges}, the last three do not depend on the time labels, either by symmetry (3, 4) or because there is  only one label on the pattern (5). Using~\cite{feuilloley2021graph}, we know that the corresponding class of graph is the permutation graphs, i.e. they are comparability graph and their complement also. Therefore this class is included in permutation graphs.
\end{proof}

It should be noted that if we consider only patterns of Figure~\ref{fig:non_edges}
  ignoring the labels of the  2 first patterns, this corresponds exactly to the particular case of trivially perfect graphs~\cite{feuilloley2021graph}. These graphs are also known as comparability graphs of trees or quasi-threshold graphs.

\section{Multi-crossing mobility model}\label{sec:multiedges}

\subsection{General recognition algorithm}

We consider now the case of multicrossing. For each edge, $\lambda$ no longer needs to be single valued. We introduce an algorithm that detects the temporal graphs resulting of multicrossing. We provide an impossibility result about characterizing this class of graphs. We give a way to check, for any triplet of nodes, if their temporal subgraph corresponds to a multicrossing sequence or not, by using an automata approach.

\begin{algorithm}

  \caption{Multicrossing recognition algorithm.}\label{alg:multicrossing-algorithm}
  
  \Input{The representation $\mathcal{R}(\mathcal{G})$ of a temporal graph $\mathcal{G}=((V,E),\lambda)$.}
  \Output{A final vertex ordering $\pi$.}
  
  Scan $\mathcal{R}(\mathcal{G})$ to obtain the list $V$ of nodes.\\
  Initialize for each node $u\in V$ a doubly linked list $L(u)$ containing $u$.\\
  We let $N_L(u)$ denote the neighbors of each $u\in V$ in its list (initially, $N_L(u)=\emptyset$).\\
  We maintain the collection $\mathcal{C}=\set{L(u) : u\in V}$ which represents a partition of $V$.\\
  For each extremity $u\in V$ of a list, we also maintain a pointer $P(u)$ to its list (initially, $P(u)$ points to $L(u)$ for all $u\in V$, when $u$ is not an extremity $P(u)$ is set to $NIL$).\\
  
  \For{$(u,v,t)\in \mathcal{R}(\mathcal{G})$}{
    \eIf{$v\in N_L(u)$}{
      \Comment{$u$ and $v$ are neighbors in the same list.}
      Exchange $u$ and $v$ in their list.\\
      Update $P(u)$ and $P(v)$).\\
    }{
      \eIf{$|N_L(u)|\le 1$ and $|N_L(v)|\le 1$ and $NIL\neq P(u)\not=P(v) \neq NIL$}{
        \Comment{$u$ and $v$ are extremities of two distinct lists.}
        Let $L_u$ and $L_v$ be the lists pointed by $P(u)$ and $P(v)$ respectively.\\
        Concatenate $L_u$ and $L_v$ into a list $L$ so that $u$ and $v$ become neighbors.\\
        $\mathcal{C}:=(\mathcal{C}\setminus\set{L_u, L_v})\cup L$.\\
        Exchange $u$ and $v$ in $L$.\\
        Update $P(u)$ and $P(v)$ to $NIL$ and for each extremity $w$ of $L$ set $P(w)$ to $L$.\\
      }{
        \Return{$\bottom$} \Comment{Failure.}
      }
    }
  }
  
  \Return{the concatenation (in any order but including separators) of all remaining lists.}

\end{algorithm}

We now propose an algorithm for recognizing if a sequence of edge appearances corresponds to a $1D$-mobility schedule. The main idea is to scan the sequence while maintaining all possible orderings of the nodes at current time through a collection $\mathcal{C}=\set{L_1,\ldots,L_k}$ of lists partitioning the set of nodes. Each list $L_i$ in the collection indicates that the nodes it contains are consecutive in the current ordering $\pi$: $L_i$ or its reverse must be a sub-sequence of $\pi$. In other words, the set of possible orderings represented by the collection $\mathcal{C}$ is made of concatenations $L'_{\sigma(1)},\ldots,L'_{\sigma(k)}$ where $\sigma$ is any permutation of $[k]$ and each $L'_i$ is either $L_i$ or its reverse. Each time two nodes $u$ and $v$ cross each other, we know that they must be consecutive in the current ordering and that their positions get exchanged. If they are in the same list, they must thus be neighbors and get exchanged in the list. Otherwise, the only possibility is that the two lists containing $u$ and $v$ are consecutive sub-sequences of the current ordering $\pi$. Moreover, $u$ and $v$ must be extremities of these sub-sequences, and they are next one to another in $\pi$. The two lists must thus be concatenated before exchanging the positions of $u$ and $v$. See Algorithm~\ref{alg:multicrossing-algorithm} for a formal description.

\begin{proposition}
Algorithm \ref{alg:multicrossing-algorithm} is correct and runs in linear time.  
\end{proposition}

\begin{proof}
The correctness of Algorithm~\ref{alg:multicrossing-algorithm} comes from the fact that $\mathcal{C}$ represents all possible orderings of the nodes after each iteration of the scanning loop. This can be proven by a simple induction based on the discussion above. As each iteration takes constant time, the whole computation takes linear time. We also use $O(1)$ space per node and the whole computation can be performed in a streaming fashion using $O(n)$ space. 
\end{proof}

The collection $\mathcal{C}$ can be seen as a PQ-tree with only two levels to make the link with \cite{Villani2021}. We use a simple representation as a collection of lists for the sake of simplicity. This algorithm can be seen as an online algorithm since if new edges appear we can reload the algorithm separating the lists.

Finally, we can also obtain a possible initial ordering by running the algorithm on the reverse of $\mathcal{R}(\mathcal{G})$ by a time reversal argument.

\subsection{Unordered forbidden patterns}

In this section, we consider unordered forbidden patterns for multicrossing graphs. This time, we no longer have ordering on the nodes for patterns. It means that we directly consider exclusion of temporal patterns, as defined in the preliminaries section.

We prove that there does not exist any finite set of forbidden (unordered) patterns to characterize the mobility graphs with multicrossing:

\begin{theorem}
    For any $k\in\mathbb{N}$, there exists a graph on $k$ nodes that does not correspond to a multicrossing such that no pattern on a subset of its nodes can be forbidden.
\end{theorem}
\begin{proof}
    We consider the set of vertices $V=\{v_1,\ldots,v_k\}$.
    The idea is that if we have a sequence of edges of the form $(x_1,x_2)^2(x_2,x_3)^2\ldots (x_{k-1},x_k)^2$, it means that those nodes are ordered $x_1,x_2,\ldots,x_k$ (or its mirror) both at the beginning and after this sequence of edges. Indeed, when we have twice in a row the edge $(x,y)$, it implies that $x$ and $y$ are next to each other, and both are in the same position as before.

    To the above sequence, by adding twice the edge $(x_1,x_k)$ at the end, we get a forbidden pattern, as both $x_1$ and $x_k$ are in the same list in Algorithm~\ref{alg:multicrossing-algorithm}, but on both ends.

    If we take a subset of size $k-1$, it means that there is some $x_i$ missing. From the initial order $x_{i+1}\ldots x_kx_1\ldots x_{i-1}$, after applying the corresponding sequence of edges, we got no issue with Algorithm~\ref{alg:multicrossing-algorithm} and end up in the same configuration that the one in the beginning.
\end{proof}

\subsection{Automaton approach}

We have noticed in the case of Mobility Cliques that a contradiction is met when 3 agents have an impossible situation. For this reason, we consider now only 3 nodes, with a new approach. Let $a$, $b$ and $c$ be three nodes starting with order $abc$. 
To see the evolution of the ordering of the 3 nodes, we consider an automaton where each state represents an ordering. For each possible crossing, we put an edge with the pair as a label, going from an ordering to its update after the switch. For example, from $abc$, with edge $ab$, we go to state $bac$. This automaton is represented in Figure~\ref{fig:automata1}.

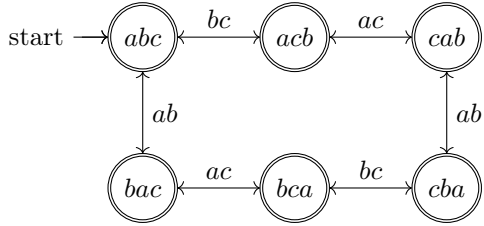
\begin{figure}[htp]
\begin{center}
\begin{tikzpicture} [node distance = 2cm, on grid]
   \tikzstyle{circlenode}=[draw,circle,minimum size=70pt,inner sep=0pt]
    \tikzstyle{whitenode}=[draw,circle,fill=white,minimum size=18pt,inner sep=0pt]
 
\node (a1) [state, initial, accepting] at (0,0) {$abc$};
\node (b1) [state, accepting] at (2,0) {$acb$};
\node (c1) [state, accepting] at (4,0) {$cab$};
\node (a2) [state, accepting] at (0,-2) {$bac$};
\node (b2) [state, accepting] at (2,-2) {$bca$};
\node (c2) [state, accepting] at (4,-2) {$cba$};
\draw (a1) edge [<->] node [above] {$bc$} (b1);
\draw (b1) edge [<->] node [above] {$ac$} (c1);
\draw (a2) edge [<->] node [above] {$ac$} (b2);
\draw (b2) edge [<->] node [above] {$bc$} (c2);
\draw (a1) edge [<->] node [right] {$ab$} (a2);
\draw (c2) edge [<->] node [right] {$ab$} (c1);

\draw (a1) edge [<-] node [above] {} (-0.8,0);
\end{tikzpicture} 
\end{center}
\caption{Automaton of possible edge sequences on the nodes $a$, $b$ and $c$.}\label{fig:automata1}
\end{figure}

The goal is to capture the possible sequences of edges of $V$. We can see the ordered sequence of edges as a word in $\Sigma^*$ with $\Sigma=\{ab,bc,ac\}$. To consider the impossible words, we need to look at the complementary of this automaton. We are doing the following shortcuts:
\begin{itemize}
    \item We can see that only the node in the middle is relevant and the two states corresponding to an ordering and its left-right mirror. We will now identify these pairs of states, merging $abc$ and $cba$ into state $b$, $acb$ and $bca$ into state $c$, $bac$ and $cab$ into state $a$.
    \item To simplify the reading, each edge will be identified with the uppercase version of the node that is not in it. For example, $ab$ becomes $C$, and $\Sigma=\{A,B,C\}$.
    \item We add one more state $\bottom$, corresponding to the final state that can be reached for leaving  the automaton from Figure~\ref{fig:automata1}.
\end{itemize}

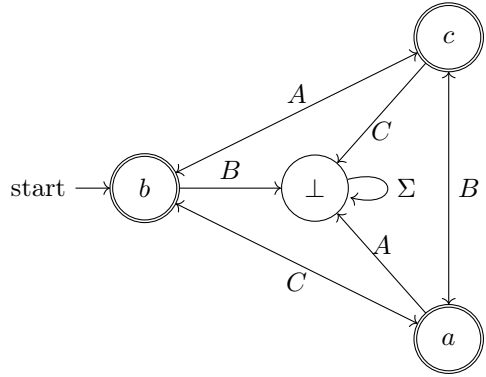
\begin{figure}[htp]
\begin{center}
\begin{tikzpicture} [node distance = 2cm, on grid]
   \tikzstyle{circlenode}=[draw,circle,minimum size=70pt,inner sep=0pt]
    \tikzstyle{whitenode}=[draw,circle,fill=white,minimum size=18pt,inner sep=0pt]
 
\node (b) [state, initial, accepting] at (0,0) {$b$};
\node (c) [state, accepting] at (4,2) {$c$};
\node (a) [state, accepting] at (4,-2) {$a$};
\node (f) [state] at (2.24,0) {$\bottom$};
\draw (b) edge [<->] node [above] {$A$} (c);
\draw (a) edge [<->] node [right] {$B$} (c);
\draw (a) edge [<->] node [below] {$C$} (b);
\draw (a) edge [->] node [above] {$A$} (f);
\draw (b) edge [->] node [above] {$B$} (f);
\draw (c) edge [->] node [below] {$C$} (f);
\draw (f) edge [->, loop right] node [right] {$\Sigma$} (f);

\end{tikzpicture} 
\end{center}
\caption{Automaton of possible edge sequences.}\label{fig:automata2}
\end{figure}

From this, we get the automaton of Figure~\ref{fig:automata2}. In particular, we consider the words ending in the final state $\bottom$. This state corresponds to the case where we had an impossible edge happening before. From there, we accept any edge from $\Sigma$.

Using usual automaton reductions, we can deduce a regular expression corresponding to the set of words on $\Sigma^*$/sequences of edges that are possible in our model. One can prove that the corresponding language $L_{MC}$ is:
\begin{equation}\label{eq:regexprposs}
L_{MC}=L_{Id}
\left(\varepsilon+AB^*+CB^*\right)
\end{equation}
where $\varepsilon$ denotes the empty word and $L_{Id}$ is the language of sequences leading back to the initial ordering with $b$ in the middle:
\begin{equation}\label{eq:regexprid}
L_{Id} = \left[A(BB)^*A+C(BB)^*C+AB(BB)^*C+CB(BB)^*A\right]^*
\end{equation}

Equivalently, our model forbids the language recognized by the automaton of Figure~\ref{fig:automata2}:
\begin{equation}\label{eq:regexpr}
L_{Id} [B+(A+CB)(BB)^*C+(C+AB)(BB)^*A]\Sigma^*
\end{equation}

\begin{theorem}\label{th:mulauto}
    Let $n$ agents on a line initially ordered as $\{a_1,\ldots,a_n\}$, and an ordered list of triplets $\mathcal{R}=(u_1,v_1,t_1),\ldots,(u_M,v_M,t_M)$ of temporal edges such that $\forall i,<j\le M$, $t_i<t_j$. We consider the corresponding temporal graph $\mathcal{G}=((V,E),\lambda)$ with $V=\{a_1,\ldots,a_n\}$.

    $\mathcal{G}$ corresponds to multicrossing if and only if, for any $i<j<k\le M$, the sequence of edges $A={a_ja_k}$, $B={a_ia_k}$ and $C={a_ia_j}$ ordered with multiplicity according to $\mathcal{R}$ is in $L_{MC}$.
\end{theorem}

\begin{proof}
    As $L_{MC}$ corresponds to the exact set of possible sequences on $A$, $B$ and $C$, multicrossing directly implies the property for any 3 distinct agents. On the other hand, one can notice that $\mathcal{G}$ does not correspond to a multicrossing if there exists $a_i$, $a_j$, $a_k$ and $t$ such that $(a_i,a_j,t)\in\mathcal{R}$ and $a_k$ is between $a_i$ and $a_j$ at time $t-1$. Consider the earliest time $t$ where this happens for a triplet $a_i$, $a_j$, $a_k$.
    
    As there is no influence from any other node in the sequence in where those three nodes are at time $t-1$ relatively to each other, we consider the crossings between these three nodes in the prefix $\mathcal{R}_{< t}=\{(u_1,v_1,t_1): t_1< t\}$.
    Let us assume $i<j<k$ (the other cases are handled in the same way, choosing the right names for the edges). Let name $A={a_ja_k}$, $B={a_ia_k}$ and $C={a_ia_j}$ for the (partial) edges, and consider the sequence of apparition of these edges in $\mathcal{R}_{< t}$. As the sequence of crossing is valid up to that point, the corresponding word on $A,B,C$ leads the automaton of Figure~\ref{fig:automata2} to one of the states $a,b,c$. As $a_k$ is in the center, this state must be $c$, and the transition $C$ at time $t$ leads to the $\bottom$ state where the automaton remains until the end of the sequence. This allows to conclude the second part of the proof.
\end{proof}
\section{Circular temporal clique}\label{sec:circular}

\subsection{Preliminary}

\emph{Circular ordering.} A circular ordering of a set $X$ is a ternary relation $C \subseteq X$ such that for any four elements $x, y, z, w \in X$ the following statements hold: \cite{guzman2023describing}
\begin{itemize}
    \item if $(x,y,z) \in C$ then $(y,z,x) \in C$,
    \item if $(x,y,z) \in C$ then $(x,z,y) \notin C$,
    \item either $(x,y,z) \in C$ or $(x,z,y) \in C$.
    \item if  $(x,y,z) \in C$ and $(x,z,w) \in C$, then $(x,y,w) \in C$, 
\end{itemize}
The elements in $X$ can be uniquely, up to rotations, arranged around a fixed circle such that, for any tuple $(x,y,z) \in C$, if we read elements in the circle clockwise from $x$, then the reading ordering is $x,y,z$. We say that $x$ and $z$ are \emph{consecutive} in $C$ when $C$ does not contain triplets $(x,y,z)$ nor $(z,y,x)$ for any $y\in X$.

\emph{Circularly ordered graphs.} A circularly ordered temporal graph is a pair $(\mathcal{G}, \pi)$, where $\mathcal{G}$ is a temporal graph and $\pi$ is a circular ordering of its nodes. Similarly, a \emph{circularly ordered temporal pattern} $(\mathcal{H}, \pi)$ is a temporal pattern $\mathcal{H}$ together with a circularly  ordering $\pi$ of its node. A circularly ordered incremental temporal graph $(\mathcal{G}, \pi')$ excludes $(\mathcal{H}, \pi)$ when it does not have any temporal subgraph $\mathcal{H'}$ which is isomorphic to $\mathcal{H}$ through an isomorphism $\phi$ preserving relative circular orderings, that is $(\phi(u),\phi(v), \phi(w)) \in \pi$ whenever $(u,v,w) \in \pi'$. We then also say that the ordering $\pi'$ excludes $(\mathcal{H}, \pi)$ from $\mathcal{G}$, or simply excludes $(\mathcal{H}, \pi)$ when $\mathcal{G}$ is clear from the context.

\subsection{Model}
We introduce here the notion of temporal graph associated to mobile agents moving along a circle that is an one-dimensional sphere.
Consider $n$ mobile agents in an oriented circle. 
At time $t_0=0$, they initially appear along the line according to a circular ordering $\pi_0$. 
These agents move along the circle and can cross one another as time goes on. We assume that a crossing is always between exactly two neighboring agents, and a single pair of agents cross each other at a single time. By ordering the crossings, we have the $k$th crossing happening at time $t_k=k$.

A \emph{circular-mobility schedule} from an ordering $\pi_0=a_1,\ldots,a_n$ of $n$ agents is a sequence $x=x_1,\ldots,x_T$ of crossings within the agents. 
Each crossing $x_t$ is a pair $uv$ indicating that agents $u$ and $v$ cross each other at time $t$. Note that the circular ordering $\pi_{t}$ at time $t$ is obtained from $\pi_{t-1}$ by permuting $u$ and $v$, and it is thus required that they appear consecutively in $\pi_{t-1}$. To such a schedule, we can associate a temporal graph $\mathcal{G}_{\pi_0,x}=((V,E),\lambda)$ such that:
\begin{itemize}
    \item $V=\{a_1,\ldots,a_n\}$,
    \item 
    $E=\{uv : \exists t\in [T], x_t=uv\}$,
    \item 
    for all $uv\in E$, $\lambda(uv)=\{t : x_t=uv\}$.
\end{itemize}

 We are interested in particular by the case where all agents cross each other exactly once as the resulting temporal graph is then a temporal clique which
 is called a \emph{circular-mobility temporal clique}. 
 More generally, we say that an incremental temporal graph $\mathcal{G}$ \emph{corresponds to a circular schedule} if there exists some circular ordering $\pi$ of its vertices and a circular schedule $x$ from $\pi$ such that the identity is an isomorphism from $\mathcal{G}$ to $\mathcal{G}_{\pi,x}$. The graph $\mathcal{G}$ is then called a \emph{circular temporal graph}.

\subsection{Characterization}

\begin{theorem}\label{th:cir}
A temporal clique is a circular-mobility temporal clique if and only if there exists a circular ordering $\pi$ of its nodes that excludes the circularly ordered temporal patterns of Figure \ref{fig:circular} as a prefix.
\end{theorem}

\begin{figure}[t]
    \input{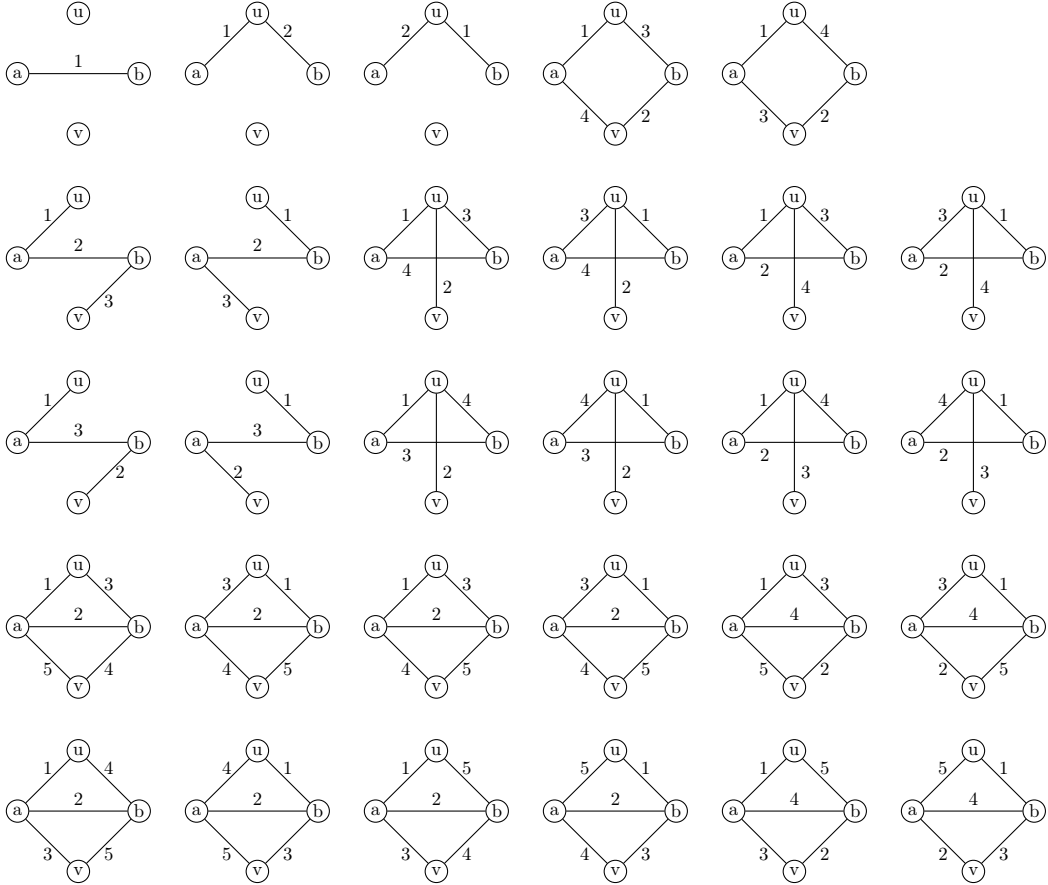}
    \caption{Circularly ordered forbidden prefix patterns in an ordered circular-mobility temporal clique. Each pattern is ordered clock-wise and associated circular ordering $a,u,b,v$.} 
        \label{fig:circular}
\end{figure}

The theorem derives from the following lemmas.

\begin{lemma}\label{lem:cir4nodes}
The relative circular ordering of four agents, $a,b,c,d$, does not change after a crossing of two agents which are both distinct from $a,b,c,d$.
\end{lemma}

\begin{lemma}\label{lem:cirpatterns}
    Let $\mathcal{G}$ be a circular-mobility temporal clique, and $\pi$ be the initial circular ordering of nodes in the oriented circle, then $\pi$ excludes the circular forbidden patterns in Figure~\ref{fig:circular} from $\mathcal{G}$.
\end{lemma}

\begin{proof}
    First observe that the relative circular ordering of four agents, $a,b,c,d$, does not change after a crossing of two agents which are both distinct from $a,b,c,d$. Furthermore, if two agents $u$ and $v$ are not consecutive at a certain time, they cannot cross each other in the next time step. One can check that each pattern of Figure~\ref{fig:circular} contains such a forbidden step.
\end{proof}

Informally, for a circular order of four agents, there are $6!$ incremental temporal cliques. We considered all of those circular patterns exhaustively, to see if the corresponding four-agent schedules with the corresponding original orderings are feasible or not. All forbidden patterns are exactly the patterns that have a prefix among the ones displayed in Figure \ref{fig:circular}.

\begin{lemma}
Any temporal clique $\mathcal{G} \in \mathcal{C}$ having a circular ordering $\pi$ excluding the circular forbidden patterns in Figure \ref{fig:circular}, can be associated to a circular-mobility schedule from $\pi$ as initial circular ordering of agents. 
\end{lemma}

\begin{proof}
Recall that, up to isomorphisms, we can assume that $\mathcal{G}$ has life time $T=n(n-1)/2$ and that exactly one edge appears at each time $t \in [T]$.

Starting from the initial circular ordering $\pi_0=\pi$, we define a sequence $\pi_1,\ldots,\pi_T$, where $T = n(n-1)/2$, of circular orderings, corresponding to what, we believe, to be the circular positions of the agents at each time step if we read the edges of $\mathcal{G}$ by increasing time label as a mobility schedule. More precisely, for each $t\in T$, $\pi_t$ is defined from $\pi_{t-1}$ as follows. Consider the edge  $uv=\lambda^{-1}(t)$ appearing at time $t$ in $\mathcal{G}$. We define $\tau_t$ as the transposition exchanging $u$ and $v$ in $w_{t-1}$. Equivalently, we set $\tau_t=(i,j)$ where $i$ and $j$ respectively denote the indexes of $u$ and $v$ in $\pi_{t-1}$, i.e. $\pi_{t-1}(i)=u$ and $\pi_{t-1}(j)=v$. We then set $\pi_t=\pi_{t-1}\tau_t$.

We need to prove that $u_t$ and $v_t$ are indeed adjacent in $\pi_{t-1}$.
Consider the first time $t$ when this fails to be. That is $\tau_1, \ldots, \tau_{t-1}$ are adjacent transpositions, edge $uv$ appears at time $t$, i.e. $uv = u_tv_t$, and $u,v$ are not consecutive in $\pi_{t-1}$. It implies that there exists two other nodes $a,b$ such that in $\pi_{t-1}$, the relative circular ordering amongs the four vertices is $a, u_t, b, v_t$. We consider the temporal subgraph, $\mathcal{H}$, induced by the four vertices $a, u_t, b, v_t$. Then $\mathcal{H}$ belongs to the set $S$ of forbidden circular patterns. 
\end{proof}

Remark. The set of forbidden prefix patterns in Figure \ref{fig:circular}  can also be used to characterize temporal mobility graphs arising in 1D model where pairs cross each other at most once. Each (partial) completion of a prefix of Figure \ref{fig:circular} cannot happen in a circular mobility clique, and conversely, any impossible pattern must have a prefix from Figure~\ref{fig:circular}.

\subsection{Multi-crossing circular mobility model}

Each state of the automaton \ref{fig:Cirautomata2} represents two circular orderings of four agents, $a,b,c,d$, in the circle. The state $x-y,z-t$ corresponds to circular orderings such that two agents $x$, $y$ are not adjacent; and two agents $z$, $t$ are not adjacent. Intuitively, from a current circular ordering corresponding to the state $x-y,z-t$, only agents, which are adjacent, can cross. For example, from an ordering corresponding to the state $a-c,b-d$, after the agents $b$ and $c$ cross, the resulting circular ordering is an ordering corresponding to the state $a-b,c-d$. An impossible crossing arises when two non-adjacent agents, in a current circular ordering, cross. This is captured by transition rules representing crossing of non-adjacent agents, which leads states to an accepting state, $\bot$. Formally, the automaton, $(Q, \Sigma, \delta, q_0, F)$, of impossible crossing-sequences is defined as follows. 

\begin{itemize}
    \item Each state, $x$-$y,z$-$t$, in $Q$, represents two circular orderings ($x,z,y,t$ and $x,t,y,z$), in which, two agents $x$, $y$ are not adjacent; and two agents $z$, $t$ are not adjacent;
    \item $\Sigma = \{ ab,ac,ad,bc,bd,cd\}$ corresponds to crossings between two agents; 
    \item Function $\delta$, which is a transition function from $Q$ to $Q$, is given by Figure~\ref{fig:Cirautomata2};
    \item Initial state, $q_0$, is the state $a$-$c,b$-$d$;
    \item Accepting state, $F$, is the state $\bot$.
\end{itemize}

Each word of the automaton corresponds to a sequence of crossings which is impossible in a circular mobility model.

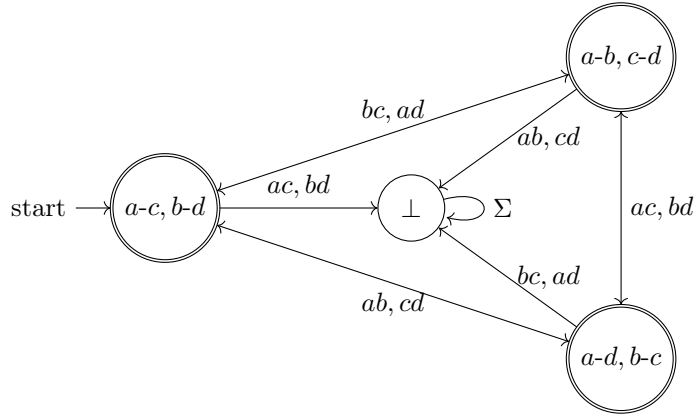
\begin{figure}[htp]
    \begin{center}
\begin{tikzpicture} [node distance = 2cm, on grid]
   \tikzstyle{circlenode}=[draw,circle,minimum size=60pt,inner sep=0pt]
    \tikzstyle{whitenode}=[draw,circle,fill=white,minimum size=18pt,inner sep=0pt]
 
\node (ac) [state, initial, accepting] at (-1,0) {$a$-$c,b$-$d$};
\node (ab) [state, accepting] at (5,2) {$a$-$b,c$-$d$};
\node (ad) [state, accepting] at (5,-2) {$a$-$d,b$-$c$};
\node (f) [state] at (2.24,0) {$\bottom$};
\draw (ac) edge [<->] node [above] {$bc,ad$} (ab);
\draw (ac) edge [<->] node [below] {$ab,cd$} (ad);
\draw (ab) edge [<->] node [right] {$ac,bd$} (ad);
\draw (ac) edge [->] node [above] {$ac,bd$} (f);
\draw (ab) edge [->] node [right] {$ab,cd$} (f);
\draw (ad) edge [->] node [right] {$bc,ad$} (f);
\draw (f) edge [->, loop right] node [right] {$\Sigma$} (f);

\end{tikzpicture} 
\end{center}
\caption{Automata of possible crossing-sequences.}\label{fig:Cirautomata2}
\end{figure}

We can see that the automaton is the same as the one in Figure~\ref{fig:automata2}, if we replace $B$ with $ac,bd$, $C$ with $ab,dc$ and $A$ with $bc,ad$. Hence, we deduce the set of impossible sequences by doing the corresponding changes in the regular expression~(\ref{eq:regexpr}).

\begin{theorem}\label{th:mulautocir}
    Let  $\mathcal{G}=((V,E),\lambda)$ be a temporal graph on $n$ vertices labeled as $V=\{a_1,\ldots,a_n\}$. We order the temporal edges of $\mathcal{G}$ as a sequence of triplets $\mathcal{R}=(u_1,v_1,t_1),\ldots,(u_M,v_M,t_M)$ such that $\forall i,<j\le M$, $t_i<t_j$. 
    If there exists a circular ordering on vertices of $\mathcal{G}$, such that any sub-sequence of $\mathcal{R}$ restricted to crossing among four agents $a,b,c,d$ is recognized by the automaton in Figure~\ref{fig:Cirautomata2}, then $\mathcal{G}$ is a multi-crossing mobility graph.
\end{theorem}

\begin{proof}
    The proof of Theorem \ref{th:mulautocir} is similar to the proof of Theorem \ref{th:mulauto}. 
\end{proof}

\section{Conclusion and perspectives}\label{sec:conclusion}

In this paper, we have introduced the first notion of forbidden patterns in temporal graphs. In particular, this notion allowed us to characterize a range of classes of temporal cliques corresponding to 1D mobility models. In case when agents cross each other exactly once, we show that the corresponding class of temporal cliques has spanners of size $2n-3$, following the conjecture from~\cite{casteigts2021temporal}. 
We note that in our model when two agents meet, they are forced  to switch positions along the line. 
As a next step, one could lift this constraint and try to characterize graphs arising from the more general model, where the adjacency between two agents do not necessarily imply that they swap places. Another possible generalization is to study mobility models in higher dimensions. It is known that (static) unit-disk graphs have infinitely many minimal forbidden induced subgraphs~\cite{AZ18}. Nevertheless, it would be interesting to study this class from the point of view of forbidden patterns.
We also propose a more constrained 2D model giving a \emph{directed} temporal graph: here vertices are agents in $\mathbb R^2$, each moving in a fixed direction and at time $t$, we have a (directed) edge from agent $i$ to agent $j$ if and only if the line through the positions of agents $i$ and $j$ is orthogonal to the moving direction of agent $i$, see Figure~\ref{fig:chicken}. We call this the ``chicken model'', imagining that each agent has eyes on each side and sees others when they are on one side or the other. Note that when all agents have parallel trajectories and move in the same direction, it is then similar to our 1D-mobility model.
\begin{figure}
    \centering
    \includegraphics[width=0.4\textwidth]{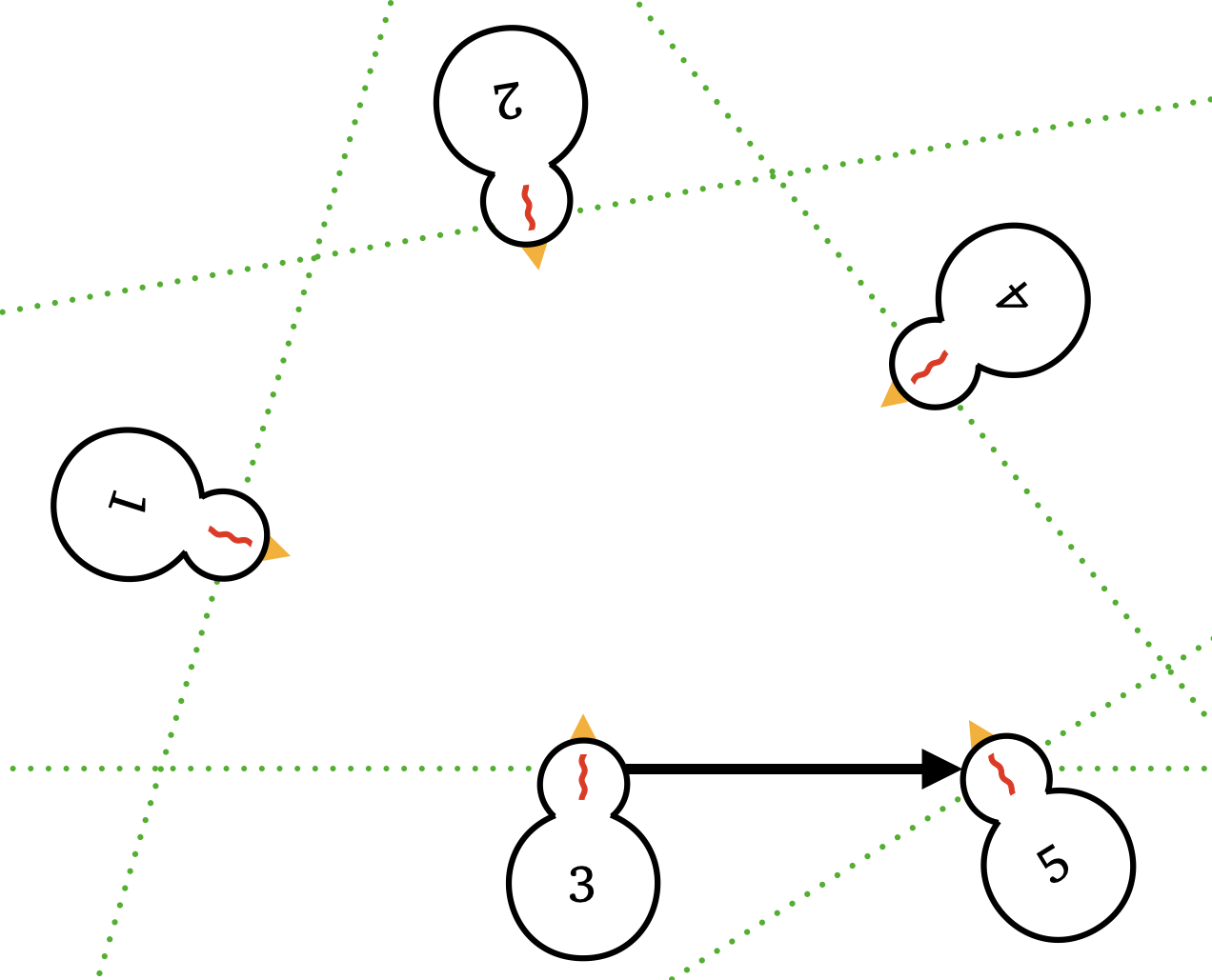}
    \caption{Illustration of the appearance of the directed edge $3\rightarrow 5$ in our 2D-model where agents, like chickens, have a visibility region orthogonal to their direction of movement.
    }
    \label{fig:chicken}
\end{figure}
More generally, can we find mobility models in higher dimensions that also provide interesting temporal cliques?

Our work can also be seen as a characterization of square integers matrices in terms of patterns (if we fill the adjacency matrix of a 1D-mobility clique with the time of appearance of each edge), perhaps it could be generalized to a study of well-structured matrices as in the seminal work of~\cite{LaurentST17} on Robinsonian matrices which are closely related to interval graphs.

\subsubsection*{Acknowledgement}\label{sec:appendix}

This work was supported by the French ANR project TEMPOGRAL (ANR-22-CE48-0001). The third author has received funding from the European Union's Horizon 2020 research and innovation program under the Marie Skłodowska-Curie grant agreement No 945332.

\bibliographystyle{splncs04}
\bibliography{ref}

\end{document}